\documentclass[11pt]{article}

\usepackage{amsmath, amssymb, amsthm}
\usepackage[square,numbers]{natbib}
\usepackage{graphicx, verbatim}
\usepackage{url}
\usepackage{multirow}
\usepackage{rotating}
\usepackage{threeparttable}
\usepackage{listings}
\usepackage{epsfig}
\usepackage{subfig}
\usepackage{eurosym}
\usepackage{float}
\usepackage{caption}
\usepackage{adjustbox}
\usepackage{tcolorbox}
\usepackage{longtable}
\usepackage{footnote}
\usepackage{enumitem}
\makesavenoteenv{tabular}
\makesavenoteenv{table}

\usepackage{bbm}
\usepackage[mathscr]{euscript}

\numberwithin{equation}{section}
\newtheorem{thm}{Theorem}  
\newtheorem{prop}[thm]{Proposition}  

\title{Model-Independent Price Bounds for Catastrophic Mortality Bonds}

\author{Raj Kumari Bahl and Sotirios Sabanis}

\date{\today}

\begin{document}

\maketitle

\begin{abstract}
\pagenumbering{roman}
     In this paper, we are concerned with the valuation of Catastrophic Mortality Bonds and, in particular, we examine the case of the Swiss Re Mortality Bond 2003 as a primary example of this class of assets. This bond was the first Catastrophic Mortality Bond to be launched in the market and encapsulates the behaviour of a well-defined mortality index to generate payoffs for bondholders. Pricing these type of bonds is a challenging task and no closed form solution exists in the literature. In our approach, we express the payoff of such a bond in terms of the payoff of an Asian put option and present a new approach to derive model-independent bounds exploiting comonotonic theory as illustrated in \cite{prime1}, \cite{2} and \cite{Simon} for the pricing of Asian options. We carry out Monte Carlo simulations to estimate the bond price and illustrate the quality of the bounds.

\bigskip


\bigskip


\bigskip


\end{abstract}

\section{Introduction}
\pagenumbering{arabic}
    In the present day world, many financial institutions face the risk of unexpected fluctuations in human mortality and clearly, this risk has two aspects. On one side, life insurers paying death benefits will suffer an economic loss if actual rates of mortality are in excess of those expected, due to catastrophic events such as a severe outbreak of an epidemic or a major man-made or natural disaster. This side of the risk is known in the literature by the name of \emph{mortality risk}. On the other hand, pension plan sponsors, as well as insurance companies providing retirement annuities, are subject to \emph{longevity risk}, that is, the risk that people outlive their expected lifetimes. For these institutions, the longer the life-span of people, the greater the period of time over which retirement income must be paid and, hence, the larger the financial liability.

An unanticipated change in mortality rates will affect all policies in force. Therefore, as opposed to the random variations between lifetimes of individuals, it cannot be diversified away by increasing the size of the portfolio. Reinsurance is one possible solution to the problem, but its capacity is usually limited. Alternatively, the risk may be naturally hedged or reduced through balancing products. For example, an insurance company may sell life insurance to the same customers who are buying life annuities.
The resulting combination would then reduce the company's exposure to future changes in mortality, consequently permitting a reduction of capital reserves held in respect of mortality or longevity risk. This idea of compensating longevity risk by mortality risk is often referred to as \emph{natural hedging}. However, this strategy, as \cite{Cox} pointed out, may be cost prohibitive and may not be practical in some circumstances.

As a result, a natural remedy to tackle these risks has emerged in the form of what is known as mortality  securitization which manifests itself in the form of \emph{mortality-linked securities} abbreviated in the literature as \emph{MLSs}. These securities, which typically come in the form of bonds, provide a tool in the hands of insurers to transfer their mortality-sensitive exposures to a vested number of investors in the capital market, offering them a risk premium in return. Mortality-linked securities differ from their longevity counterparts in the sense that while the former have their cash flows linked to a mortality index, the latter are based upon a survivor index. For a more detailed review of the two type of bonds, one can refer to \cite{Melnick}.
In fact mortality-linked securities are also known as \emph{Extreme Mortality Bonds} or \emph{EMBs} or \emph{Catastrophe (CAT) Mortality Bonds} or \emph{CATM bonds} since they are triggered by a catastrophic evolution of death rates of one or more populations. These bonds may be appealing to the investors because of their potential of providing diversification to the portfolio. The return on these bonds generally does not bear any correlation with the return on other investments, such as fixed income or equities. From the point of view of the (re)insurer these instruments act as `Alternative Risk Transfer' (ART) mechanisms.

The pioneering MLS was the Swiss Re mortality bond (Vita I) issued in 2003 which is the prime focus of this paper. This was followed up by the EIB/BNP longevity bond issued in 2004 (\cite{Blake}; \cite{Lane}). For the former, the principal of the bond would have been reduced if there had been a catastrophic mortality event during the life of the bond, therefore allowing Swiss Re to reduce some of its exposure to extreme mortality risk. On the contrary, the latter was a 25-year longevity bond, which was intended for UK pension funds with exposures to longevity risk. This bond took the form of an annuity bond with annual coupon payments tied to the realized survival rates for some English and Welsh males. However, it did not get the same reception as the Swiss Re bond. Swiss Re followed up the success of VITA I by launching five more series of VITA bonds with the latest one being VITA VI which will cover extreme mortality events in Australia, Canada and the UK over a 5 year term from January 2016. Apart from this Swiss Re also experimented with a multi-peril bond called ``Mythen Re" which synthesized catastrophe and mortality risks, obtaining 200 million US dollars in protection for North Atlantic hurricane and UK extreme mortality risk. Many other reinsurance giants such as Scottish Re and Munich Re have also issued a score of other mortality bonds. We refer readers to \cite{Blake2008}, \cite{Coughlan}, \cite{Zhou} and \cite{Chen3} for further details. In fact it is interesting to note that Swiss Re has also launched an innovative `Longevity Trend Bond' called the Swiss Re Kortis bond in December 2010. Interested readers can refer to \cite{Chen3} and \cite{Hunt}. A more up to date list of developments connected to mortality and longevity securities and markets can be found in \cite{Tan} and \cite{Liu}. 

As an aftereffect of these innovative securities, a number of valuation approaches on MLS's have germinated.  \cite{Huang} classify the approaches into the following four, not mutually exclusive, heads:

\begin{itemize}
\item \emph{Risk-adjusted process or no-arbitrage pricing:} Under this approach, the first step is to estimate the distribution of future mortality rates under the real-world probability measure. Then the real-world distribution is transformed to its risk-neutral counterpart, on the basis of the actual prices of mortality-linked securities observed in the market. Finally, the price of a mortality-linked security
can be calculated by discounting, at the risk-free interest rate, its expected payoff under the identified risk-neutral probability measure. An important point underlying this approach is that it takes into account the actual prices. The need of market prices makes the implementation of this approach difficult. One way to effectively use the no arbitrage approach is to use a stochastic mortality model, which is, at the very beginning, defined in the real-world measure and fitted to past data. The model is then calibrated to market prices, yielding a risk-neutral mortality process from which security prices are calculated. For instance, \cite{Cairns2006} calibrate a two-factor mortality model to the price of the BNP/EIB longevity bond.
\end{itemize}

\begin{itemize}
\item \emph{The Wang transform:} It is the approach given by \cite{Wang1}, \cite{Wang2} which consists of employing a distortion operator that transforms the underlying distribution into a risk-adjusted distribution and the MLS price is the expected value under the risk-adjusted probability discounted by risk-free rate. The Wang transform was first employed for mortality-linked securities by \cite{Lin2}, and subsequently by other researchers including \cite{Dowd} and \cite{Denuit2007}. Based on the positive dependence characteristic of the mortality in catastrophe areas, \cite{Shang09} develop a pricing model for catastrophe mortality bonds with comonotonicity and a jump-difusion process. Pointing out  there is no unique risk-neutral probability in this incomplete market settings, they use the Wang transform method to price the bond. Unless a very simple mortality model is assumed, parameters in the distortion operator are not unique if we are not given sufficient market price data. For example, when \cite{Chen2009} used their extended Lee-Carter model with transitory jump effects to price a mortality bond, they were required to estimate three parameters in the Wang transform. To solve for these three parameters, Chen and Cox assumed that they were equal, but such an assumption is not easy to justify. In fact \cite{Pelsser} has questioned the Wang transform by stating that it is not a universal pricing measure for financial and insurance pricing. The Wang transform is superseded by the Esscher-Girsanov transform introduced by \cite{Goovaerts1}. Contrary to the Wang transform, the Esscher-Girsanov transform is consistent with arbitrage-free financial and insurance pricing. For more details one can refer to \cite{Goovaerts1} and \cite{Lauschagne}.
\end{itemize}

\begin{itemize}
\item \emph{Instantaneous Sharpe Ratio:} \cite{Milevsky} and \cite{Milevsky2} advocate the use of Instantaneous Sharpe Ratio for financial valuation of Mortality Risk. They assume that the company issuing a mortality-contingent claim requires compensation for this risk in the form of a pre-specified instantaneous Sharpe ratio. According to them pricing mortality contingent claims in an incomplete market via Instantaneous Sharpe Ratio has many desirable properties, which makes this method useful for pricing risks in other incomplete markets too. For more details the reader is referred to the aforesaid papers and \cite{Milevsky3} and \cite{Young7}.
\end{itemize}


\begin{itemize}
\item \emph{The utility-based valuation:} The utility based method defines an investor's utility function and maximizes an agent's expected utility subject to wealth constraints to obtain the equilibrium price of the mortality linked security. For an elaborate discussion one can review \cite{Tsai}, \cite{Cox10}, \cite{Hainaut} and \cite{Dahl}. \cite{Cox10} and \cite{Hainaut} employ an exponential utility function to compute the price of MLSs.
\end{itemize}
Apart from the aforesaid methods \cite{Beelders} and \cite{Chen10} use the extreme value theory to measure mortality risk of the 2003 Swiss Re Bond. For an interesting summary of other methods to price MLS's one can refer to \cite{Shang}, \cite{Zhou13}, \cite{Tan} and \cite{Liu}.

The methods available in literature for the pricing of MLS's offer only a limited application due to restrictions such as availability of price information or specific utility functions. The difficulty in pricing MLS's stems from the fact that the MLS market is incomplete as the underlying mortality rates are usually untradeable in financial markets. As a result, the usual no-arbitrage pricing method can only provide a price range or a price bound, instead of a single value.

Surprisingly, mortality linked securities, apart from their present day form seem to have a long history. In the 17th and 18th centuries, so-called \emph{`tontines'}, which were named after the Neapolitan banker Lorenzo Tonti, had been offered by several governments (\cite{Weir}). Within these schemes, investors made a one-time payment, and annual dividends were distributed among the survivors. Hence, while still relying on the investor's survival, his payoffs were connected to the mortality experience among the pool of subscribers. These issues were particularly successful in France, but due to high interest payments, they soon became precarious for the crown's financial situation (see \cite{Jennings}). However,
this was not only the case with tontines; life annuities, which presented another large share of the royal debt, were also offered at highly favourable conditions from the investors' perspective. This carelessness was exploited by the Genevan entrepreneur Jacob Bouthillier Beaumont in the scheme attributed
to him (cf. \cite{Jennings}). Here, annuities were subscribed on the lives of a group of Genevan girls for the account of Genevan investors. Thus, their payoffs were directly linked to the survival of the Genevan ``madmoiselles", and due to the ``generous" assumptions of the French authorities, the schemes were initially highly profitable for the Genevans, the real victim being the French taxpayer. These speculations came to an abrupt end with the French Revolution in 1789, for which the budgetary crises caused by the careless borrowing was, undoubtedly, one major reason. Until the beginning of this century, there has not been another public issue of a mortality linked security, however, there are indications of recent private transactions resembling the tontine scheme (see \cite{Dowd}). For a more detailed overview of the history of mortality contingent securities the reader is referred to \cite{Bauer} and \cite{Luis}.

Today, all around the world, investment banks and other financial service providers are working on the idea of trading longevity risk, and the first mortality trading desks have been installed solidifying that ``betting on the time of death is set”.\footnote{The Business, 08/15/2007, ``Betting on the time of death is set”, by P. Thornton.}

This paper is concerned with finding price bounds for the Swiss Re mortality catastrophe bond by expressing its payoff in the form of an Asian put option and using the theory of comonotonicity. Such a methodology has been advocated by \cite{Simon}, \cite{2} and \cite{prime1} to find a price range for Asian options. For more details on comonotonicity and its applications, one can refer to \cite{Upper} and \cite{1}.

As the MLS market is incomplete, a unique pricing measure does not exist. However, since the market is arbitrage-free, at least one risk-neutral measure can be found, that can then be used to find fair prices of mortality contingent securities.  The existence of such a measure allows us to proceed with the fair pricing of mortality contingent securities and no matter what the choice of such a measure is, the pricing is done in a model independent way. We exploit this fact and work in a model-independent setting: that is, we do not assume that the mortality evolution process behaves according to a given model, but aim to draw conclusions that hold under any model. This is in contrast to the standard approach to pricing mortality contingent products which is to postulate a model and to determine the price of the underlying as the suitably discounted risk neutral expectation of the payoff under that model. A major problem with this approach is that no model can capture the real world behaviour of MLSs fully, thus exposing the entire procedure to model risk.

The rest of this paper is organised as follows: the next section describes the structure of the Swiss Re Bond and expresses its payoff in the form of an Asian put option. Section 3 shows derivations of the lower bound for the aforesaid bond using comonotonicity. In Section 4, we use the same to derive upper bounds for the Swiss Re Bond. In Section 5, we illustrate the computation of bounds by choosing specific models for mortality index. Section 6 portrays numerical results for the derived theory and compares the results with Monte Carlo estimates of the bond price. Appropriate figures that highlight comparisons among the bounds have also been furnished. The concluding section presents conclusions and avenues for further research.

\section{Design of the Swiss Re Bond}
As pointed out in the introduction, the financial capacity of the life insurance industry to pay catastrophic death losses from natural or man-made disasters is limited. To expand its capacity to pay catastrophic mortality losses, Swiss Re procured about 400 million in coverage from institutional investors in lieu of its first pure mortality security. The reinsurance giant issued a three year bond in December 2003 with maturity on January 1, 2007. To carry out the transaction, Swiss Re set up a special purpose vehicle (SPV) called Vita Capital Ltd. This enabled the corresponding cash flows to be kept off Swiss Re's balance sheet. The principal is subject to mortality risk which is defined in terms of an index $q_{t_{i}}$ in year $t_{i}$. This mortality index was constructed as a weighted average of mortality rates (deaths per 100,000) over age, sex (male 65\% and female 35\%) and nationality (US 70\%, UK 15\%, France 7.5\%, Italy 5\% and Switzerland 2.5\%) and is given below.
\begin{equation}\label{2.0}
q_{t_{i}} = \sum_{j}C_{j}\sum_{k}A_{k}\left(G^{m}q_{k,j,t_{i}}^{m}+G^{f}q_{k,j,t_{i}}^{f}\right)
\end{equation}
where $q_{k,j,t_{i}}^{m}$ and $q_{k,j,t_{i}}^{f}$ are the respective mortality rates (deaths per 100,000) for males and females in the age group $k$ for country $j$, $C_{j}$ is the weight attached to country $j$, $A_{k}$ is the weight attributed to age group $k$ (same for males and females) and $G^{m}$ and $G^{f}$ are the gender weights applied to males and females respectively.

The Swiss Re bond was a principal-at-risk bond. If the index $q_{t_{i}}$ ($t_{i}$ = 2004, 2005 or 2006 for $i=1,2,3$ respectively) exceeds $K_{1}$ of the actual 2002 level, $q_{0}$, then the investors will have a reduced principal payment. The following equation describes the principal loss percentage, in year $t_{i}$:
\begin{equation}\label{2.1}
L_{i}=\begin{cases}
0 & \text{if } q_{t_{i}}\leq K_{1}q_{0}\\
\frac{\left(q_{t_{i}}-K_{1}q_{0}\right)}{\left(K_{2}-K_{1}\right)q_{0}} & \text{if } K_{1}q_{0}<q_{t_{i}}\leq K_{2}q_{0}\\1 & \text{if }q_{t_{i}}> K_{2}q_{0}\end{cases}
\end{equation}
In particular, for the case of Swiss Re Bond, $K_{1}=1.3$ and $K_{2}=1.5$. In lieu of having their principal at risk, investors received quarterly coupons equal to the three-month U.S. LIBOR plus 135 basis points. There were 12 coupons in all with a coupon value of
\begin{equation}\label{2.1a}
CO_{j}=\begin{cases}
\left(\frac{SP+LI_{j}}{4}\right).C & \text{if } j=\frac{1}{4},\frac{2}{4},...,\frac{11}{4},\\
\left(\frac{SP+LI_{j}}{4}.C+X_{T}\right) & \text{if } j=3,\end{cases}
\end{equation}
where $SP$ is the spread value which is 1.35\%, $LI_{j}$ are the LIBOR rates, $C=\$400$ million, $T=t_{3}$ and $X_{T}$ is a random variable representing the proportion of the principal returned to the bondholders on the maturity date such that
\begin{equation}\label{2.2}
X_{T}=C\left(1-\sum_{i=1}^{3}L_{i}\right)^{+},
\end{equation}
where $\sum_{i=1}^{3}L_{i}$ is the aggregate loss ratio at $t_{3}$. However, there was no catastrophe during the term of the bond. The discounted cash flow (DC) of payments is given by
\begin{equation}\label{2.2a}
DC\left(r\right)=\sum_{i=1}^{12}\frac{CO_\frac{i}{4}}{\left(1+\frac{r}{4}\right)^{i}}
\end{equation}
where $r$ is the nominal annual interest rate.

Further define
\[
Y_{T}=-{\displaystyle \int_{0}^{T}}\rho\left(t\right)dt
\]
where $\rho(t)$ is the US LIBOR at time $t$. As a result, the risk-neutral value at time 0 of the random principal returned at the termination of the bond is
\[
P=\mbox{E\ensuremath{_{Q}}\ensuremath{\left[e^{-Y_{T}}X_{T}\right]}}
\]
where $Q$ is the risk-neutral measure. However, under the assumption of independence of $Y_{T}$ and $X_{T}$, this reduces to
\[
P=\mbox{E\ensuremath{_{Q}}\ensuremath{\left[e^{-Y_{T}}\right]}}\mbox{E\ensuremath{_{Q}}\ensuremath{\left[X_{T}\right]}}
\]
The conditions under which it is possible (or not) to transfer the independence assumption from the physical world measure $\mathbb{P}$ to $Q$ have been discussed extensively in \cite{Dhaene}. Henceforth, in this incomplete market, we choose to price under a risk neutral measure that preserves independence between market and mortality risks. In order to proceed, we represent $\mbox{E\ensuremath{_{Q}}\ensuremath{\left[e^{-Y_{T}}\right]}}$ as $e^{-rT}$, which implies
\begin{equation}\label{2.3}
P=e^{-rT}\mbox{E\ensuremath{_{Q}}\ensuremath{\left[X_{T}\right]}}
\end{equation}
where $r$ is the risk-free rate of interest. In subsequent writing, we drop $Q$ from the above expression.

\subsection{The Principal Payoff of Swiss Re Bond as that of an Asian-type Put Option}
In fact, we can write $X_{T}$ given in \eqref{2.2} in a more compact form similar to the payoff of the Asian put option as shown below:
\begin{equation}\label{2.4}
X_{T}=D\left({q_{0}-\displaystyle \sum_{i=1}^{3}}5\left(q_{t_{i}}-1.3q_{0}\right)^{+}\right)^{+}
\end{equation}
with
\begin{equation}\label{2.5}
D=\frac{C}{q_{0}}
\end{equation}
and the strike price equal to $q_{0}$. For the sake of simplicity, we use $q_{i}$ in place of $q_{t_{i}}$ and define
\begin{equation}\label{2.6}
S_{{i}}=5\left(q_{i}-1.3q_{0}\right)^{+}
\end{equation}
and
\begin{equation}\label{2.7}
S=\displaystyle \sum_{i=1}^{3}S_{i}
\end{equation}
Using \eqref{2.6}-\eqref{2.7} in \eqref{2.4} and plugging the result into \eqref{2.3}, we have:
\begin{equation}\label{2.8}
P=De^{-rT}\mbox{E\ensuremath{\left[\left(q_{0}-S\right)^{+}\right]}}
\end{equation}
It is naturally assumed that the inequalities $S \geq q_{0}$ almost surely (a.s.) and $S \leq q_{0}$ a.s. do not hold, otherwise the problem has a trivial solution. This means that $q_{0} \in \left(F_{S}^{-1+}\left(0\right), F_{S}^{-1}\left(1\right)\right)$, where as in \cite{1}, $F_{X}^{-1}$ is the generalized inverse of the cumulative distribution function (cdf), i.e.,
\begin{equation}\label{2.5a}
F_{X}^{-1}\left(p\right)=\inf \{x \in\mathbb{R}|F_{X}\left(x\right)\geq p\},\;\;p \in \left[0,1\right]
\end{equation}
and $F_{X}^{-1+}$ is a more sophisticated inverse defined as
\begin{equation}\label{2.5b}
F_{X}^{-1+}\left(p\right)=\sup \{x \in\mathbb{R}|F_{X}\left(x\right)\leq p\},\;\;p \in \left[0,1\right].
\end{equation}
Our interest lies in the calculation of reasonable bounds for $P$. We invoke Jensen's inequality for computing the lower bounds and present our findings in the subsequent sections. We exploit this inequality twice and note that in order to maintain uniformity of having a convex function at each step, it is beneficial to consider the call counterpart of the payoff of Swiss Re Bond rather than \eqref{2.8}. We nomenclate this payoff as $P_{1}$, i.e., we have
\begin{equation}\label{2.9}
P_{1}=De^{-rT}\mbox{E\ensuremath{\left[\left(S-q_{0}\right)^{+}\right]}}
\end{equation}
We then exploit the put-call parity for Asian options to achieve the bounds for the payoff in question.

\subsection{Put-Call Parity for the Swiss Re Bond}
We now derive the put-call parity relationship for the Swiss Re Bond. For any real number $a$, we have:
\begin{equation}\label{2.10}
\left(a\right)^{+}-\left(-a\right)^{+}=a
\end{equation}
So we obtain
\[
e^{-rT}\left(\sum_{i=1}^{3}S_{i}-q_{0}\right)^{+}-e^{-rT}\left(q_{0}-\sum_{i=1}^{3}S_{i}\right)^{+}=e^{-rT}\left(\sum_{i=1}^{3}S_{i}-q_{0}\right).
\]
On taking expectations on both sides, we obtain
\[
e^{-rT}\mbox{E}\left[\left(\sum_{i=1}^{3}S_{i}-q_{0}\right)^{+}\right]-e^{-rT}\mbox{E}\left[\left(q_{0}-\sum_{i=1}^{3}S_{i}\right)^{+}\right]=e^{-rT}\mbox{E}\left[\sum_{i=1}^{3}S_{i}-q_{0}\right].
\]
Finally, on multiplying by $D$ and expanding the definition of $S_{i}$, we have
\[
P_{1}-P=De^{-rT}\mbox{E}\left[\sum_{i=1}^{3}5\left(q_{i}-1.3q_{0}\right)^{+}-q_{0}\right]
\]
\begin{equation}\label{2.11}
\Rightarrow P_{1}-P=De^{-rT}\left[5\sum_{i=1}^{3}e^{rt_{i}}C\left(1.3q_{0},t_{i}\right)-q_{0}\right],
\end{equation}
where $C\left(K,t_{i}\right)$ denotes the price of a European call on the mortality index with strike $K$, maturity $t_{i}$ and current mortality value $q_{0}$. This option would be in-the-money if the mortality index is more than $1.3q_{0}$ which is the trigger level of Swiss Re bond. Clearly, such instruments are not available for trading in the market at present. But a more complete life market is in the making and we feel such securities will soon be introduced (c.f. \cite{Blake3} and \cite{Blake2008}). The pay-off structures, i.e. the design of the issued securities and the mortality contingent payments should be developed to appear attractive to investors and the re-insurer. Although, the Swiss Re bond was fully subscribed and press reports highlight that investors were quite satisfied with it (e.g. \emph{Euroweek}, 19 December 2003), the market for mortality linked securities still needs innovations such as vanilla options on mortality index to provide flexible hedging solutions. Investors of the Swiss Re bond included a large number of pension funds as they could view this bond as a powerful hedging instrument. The underlying mortality risk associated with the bond is correlated with the mortality risk of the active members of a pension plan. If a catastrophe occurs, the reduction in the principal would be offset by reduction in pension liability of these pension funds. Moreover, the bond offers a considerably higher return than similarly rated floating rate securities (c.f. \cite{Blake}). In a manner similar to \cite{Bauer}, we feel the success of the life market hinges upon flexibility. As a result, such option-type structures enable re-insurer to keep most of the capital while at the same time being hedged against catastrophic mortality situation. \cite{Cox3} present an interesting note on the trigger level of $1.3q_{0}$ in context of 2004 tsunami in Asia and Africa. A mortality option of the above type would become extremely useful in such a case. \cite{Tsai} and \cite{Cheng} decompose the terminal payoff of the Swiss Re bond into two call options.

Equation \eqref{2.11} gives the required put-call parity relation between the Swiss Re mortality bound and its call counterpart. Define
\begin{equation}\label{2.13}
G=De^{-rT}\left[5\sum_{i=1}^{3}e^{rt_{i}}C\left(1.3q_{0},t_{i}\right)-q_{0}\right].
\end{equation}
Clearly, if we bound $P_{1}$ by bounds $l_{1}$ and $u_{1}$, then the corresponding bounds for the Swiss Re mortality bond are as follows
\begin{equation}\label{2.14}
\left(l_{1}-G\right)^{+} \leq P \leq \left(u_{1}-G\right)^{+}.
\end{equation}

\section{Lower Bounds for the Swiss Re Bond}
We now proceed to work out appropriate lower bounds for the terminal value of the principal paid in the Swiss Re Bond. For this, we first calculate bounds for the following Asian-type call option
\begin{equation}\label{4.1.1}
P_{1}=De^{-rT}\mbox{E}\ensuremath{\left[\left({\displaystyle \sum_{i=1}^{n}}S_{i}-q_{0}\right)^{+}\right]}
\end{equation}
with $T=t_{n}$ and $n=3$. The interval $\left[0,T\right]$ consists of the monitoring times $t_{1}, t_{2},...,t_{n-1}$. The undercurrent of the theory presented in this section is the paper by \cite{prime1}. In an attempt to estimate the value of the Asian call option, the authors derive four lower bounds namely trivial, $LB_{1}$, $LB_{t}^{\left(1\right)}$ and $LB_{t}^{\left(2\right)}$, which are sharper in increasing order in sense of their proximity to the actual value of the Asian call. The underlying assumption they make in deriving these bounds is that European call prices with arbitrary strikes and maturities are available in the market. Although, as our previous discussion indicates, such securities with the underlying as the mortality index have not appeared on the horizon as yet, but would be indispensable for the development of a complete life market. The first step towards designing of such securities is the need for a benchmark longevity index. The formation of Life and Longevity Markets Association (LLMA) in 2010 was an important milestone in this direction. The LLMA promotes the development of a liquid trading market in longevity and mortality-related risk, of the type that exists for Insurance Linked Securities (ILS) and other large trend risks like interest rates and inflation. There have been a few mortality indices created by various parties but we still lack a benchmark. \cite{Menioux} throws light on various longevity indices.

Invoking Jensen's inequality and conditioning on an arbitrary random variable $\Lambda$, we have
\begin{eqnarray}\label{4.1.3}
\mbox{E}\ensuremath{\left[\left({\displaystyle \sum_{i=1}^{n}}S_{i}-q_{0}\right)^{+}\right]}  & \geq & \mbox{E}\ensuremath{\left[\left(5{\displaystyle \sum_{i=1}^{n}}\left(\mbox{E}\left(q_{i}|\Lambda\right)-1.3q_{0}\right)^{+}-q_{0}\right)^{+}\right]}.
\end{eqnarray}
The general derivation concerning lower bounds for stop loss premium of a sum of random variables based on Jensen's inequality can be found in \cite{Simon} and for its application to arithmetic Asian options, one can refer to \cite{2}. We now define
\begin{equation}\label{4.1.4}
Z_{i}=5\left(\mbox{E}\left(q_{i}|\Lambda\right)-1.3q_{0}\right)^{+}; i=1,2,...,n.
\end{equation}
As a result in \eqref{4.1.3}, we have obtained
\begin{equation}\label{4.1.5}
\mbox{E}\ensuremath{\left[\left({\displaystyle \sum_{i=1}^{n}}S_{i}-q_{0}\right)^{+}\right]} \geq \mbox{E}\ensuremath{\left[\left({\displaystyle \sum_{i=1}^{n}}Z_{i}-q_{0}\right)^{+}\right]}.
\end{equation}

On investigating the relationship between $\mbox{E}\ensuremath{\left[{\displaystyle \sum_{i=1}^{n}}S_{i}\right]}$ and $\mbox{E}\ensuremath{\left[{\displaystyle \sum_{i=1}^{n}}Z_{i}\right]}$, we find that
\begin{eqnarray}\label{4.1.6}
\mbox{E}\ensuremath{\left[{\displaystyle \sum_{i=1}^{n}}S_{i}\right]}
   & \geq & \mbox{E}\ensuremath{\left[{\displaystyle \sum_{i=1}^{n}}Z_{i}\right]}.
\end{eqnarray}
On lines of \eqref{2.7}, define
\begin{equation}\label{4.1.8}
Z=\displaystyle \sum_{i=1}^{n}Z_{i}
\end{equation}
so that we can rewrite \eqref{4.1.5} as
\begin{equation}\label{4.1.10}
\mbox{E\ensuremath{\left[\left(S-q_{0}\right)^{+}\right]}}\geq \mbox{E\ensuremath{\left[\left(Z-q_{0}\right)^{+}\right]}}.
\end{equation}
In fact, the two sides of the inequality in \eqref{4.1.10} are essentially the stop-loss premiums of $S$ and $Z$. Thus, we have obtained
\begin{equation}\label{4.1.11}
S \geq_{sl} Z
\end{equation}
or
\[
S\geq_{\mbox{sl }}{\displaystyle \sum_{i=1}^{n}}\left(\mbox{E}\left(q_{i}|\Lambda\right)-1.3q_{0}\right)^{+}.
\]
Now, suitably tailoring the inequality \eqref{4.1.10} to suit our need of the Asian-type call option by multiplying by the discount factor at time $T$, we obtain
\begin{equation}\label{4.1.12}
P_{1}\geq De^{-rT}\mbox{E}\ensuremath{\left[\left({\displaystyle \sum_{i=1}^{n}}5\left(\mbox{E}\left(q_{i}|\Lambda\right)-1.3q_{0}\right)^{+}-q_{0}\right)^{+}\right].}
\end{equation}
To exploit the theory of comonotonicity see for example in \cite{1}, we now have to show that the lower bound for $S$, can be formulated as a sum of stop-loss premiums. This task becomes trivial if we can choose the conditioning variable $\Lambda$ in such a way that $\mbox{E}\left(q_{i}|\Lambda\right)$ is either increasing or decreasing for every $i$, so that the vector: $\textbf{q}{}^{\textbf{l}}=\left(\mbox{E}\ensuremath{\left(q_{1}|\Lambda\right)},\ldots,\mbox{E}\ensuremath{\left(q_{n}|\Lambda\right)}\right)$ is comonotonic. This automatically implies that the vector: $\textbf{Z}{}^{\textbf{l}}=\left(Z_{1},\ldots,Z_{n}\right)$ is comonotonic. From this point onwards, we assume that $q_{0} \in \left(F_{Z}^{-1+}\left(0\right), F_{Z}^{-1}\left(1\right)\right)$ which is not at all a restriction for all practical purposes as pointed out in section 2.1. As a result on using comonotonicity, we have
\begin{eqnarray}\label{4.1.15}
\mbox{E}\ensuremath{\left[\left(S-q_{0}\right)^{+}\right]} & \geq & {\displaystyle \sum_{i=1}^{n}}\mbox{E}\ensuremath{\left[\left(Z_{i}-F_{Z_{i}}^{-1}\left(F_{Z}\left(q_{0}\right)\right)\right)^{+}\right]}\nonumber\\
& {} & {} \;\;\;\;\;\;-\left(q_{0}-F_{Z}^{-1}\left(F_{Z}\left(q_{0}\right)\right)\right)\left(1-F_{Z}\left(q_{0}\right)\right).
\end{eqnarray}
In case if the marginal cdfs $F_{Z_{i}}$ are strictly increasing, we have the following compact expression
\begin{eqnarray}\label{4.1.15b}
\mbox{E}\ensuremath{\left[\left(S-q_{0}\right)^{+}\right]} & \geq & {\displaystyle \sum_{i=1}^{n}}\mbox{E}\ensuremath{\left[\left(Z_{i}-F_{Z_{i}}^{-1}\left(F_{Z}\left(q_{0}\right)\right)\right)^{+}\right]},\nonumber\\
& {} & {} \;\;\;\;\;\;\;\;\;\;q_{0} \in \left(F_{Z}^{-1+}\left(0\right), F_{Z}^{-1}\left(1\right)\right).
\end{eqnarray}
Note from \eqref{4.1.4} and \eqref{4.1.8} that the $Z_{i}'s$ and subsequently $Z$ are non-negative and this automatically implies $q_{0} \geq 0$. Further, by the definition of cdf, we have
\begin{equation}\label{4.1.16}
F_{Z}\left(q_{0}\right)=\textbf{P}\left[Z\leq q_{0}\right]=\textbf{P}\left[{\displaystyle \sum_{j=1}^{n}Z_{j}}\leq q_{0}\right]=\textbf{P}\left[{\displaystyle \sum_{j=1}^{n}5\left(\mbox{E}\left(q_{j}|\Lambda\right)-1.3q_{0}\right)^{+}}\leq q_{0}\right].
\end{equation}
Thus we have been able to obtain a stop-loss lower bound for $S=\sum_{i=1}^{n}S_{i}$ by conditioning on an arbitrary random variable $\Lambda$, i.e.,
\begin{equation}\label{4.1.17}
P_{1}\geq De^{-rT}{\displaystyle \sum_{i=1}^{n}}\mbox{E}\ensuremath{\left[\left(5\left(\mbox{E}\left(q_{i}|\Lambda\right)-1.3q_{0}\right)^{+}-F_{Z_{i}}^{-1}\left(F_{Z}\left(q_{0}\right)\right)\right)^{+}\right]}-K_{1},
\end{equation}
where
\begin{equation}\label{4.1.17a}
K_{1}=De^{-rT}\left(q_{0}-F_{Z}^{-1}\left(F_{Z}\left(q_{0}\right)\right)\right)\left(1-F_{Z}\left(q_{0}\right)\right).
\end{equation}

\subsection{The Trivial Lower Bound}
In case, if the random variable $\Lambda$ is independent of the mortality evolution $\left\{ q_{t}\right\} _{t\geq0}$, the bound in \eqref{4.1.12} simply reduces to:\begin{equation}\label{4.2.1}
P_{1}\geq De^{-rT}\mbox{E}\ensuremath{\left[\left({\displaystyle \sum_{i=1}^{n}}5\left(\mbox{E}\left(q_{i}\right)-1.3q_{0}\right)^{+}-q_{0}\right)^{+}\right]}
\end{equation}
or even more precisely as the outer expectation is redundant
\begin{equation}\label{4.2.2}
P_{1}\geq De^{-rT}\ensuremath{\left({\displaystyle \sum_{i=1}^{n}}5\left(\mbox{E}\left(q_{i}\right)-1.3q_{0}\right)^{+}-q_{0}\right)^{+}.}
\end{equation}
Under the assumption of the existence of an Equivalent Martingale Measure (EMM), Q, the discounted mortality process is a martingale, so that
\begin{equation}\label{4.2.4}
\mbox{E}\left[q_{t}\right]=q_{0}e^{rt}.
\end{equation}
If we substitute this in equation \eqref{4.2.2}, we obtain a very rough lower bound for the Asian-type call option
\begin{equation}\label{4.2.5}
P_{1}\geq Ce^{-rT}\left({\displaystyle \sum_{i=1}^{n}5\left(e^{rt_{i}}-1.3\right)^{+}}-1\right)^{+}=:\mbox{ lb}_{0}.
\end{equation}
In the light of put-call parity derived in section 2, the trivial lower bound for the Swiss Re mortality bond is given as
\begin{equation}\label{4.2.6}
P\geq \left(\mbox{ lb}_{0}-G\right)^{+}=:\mbox{ SWLB}_{0}.
\end{equation}
where G is defined in \eqref{2.13}.

\subsection{The Lower Bound $\mbox{SWLB}_{1}$}
To improve upon the trivial lower bound, we choose $\Lambda=q_{1}$ in \eqref{4.1.17}. Using the martingale argument for the discounted mortality process
\[
\mbox{E}\left[q_{i}|q_{1}\right]=\mbox{E}\left[e^{rt_{i}}e^{-rt_{i}}q_{i}|q_{1}\right]=e^{r\left(t_{i}-t_{1}\right)}q_{1}
\]
and so from \eqref{4.1.4}
\begin{equation}\label{4.3.2a}
Z_{i}=5\left(e^{r\left(t_{i}-t_{1}\right)}q_{1}-1.3q_{0}\right)^{+}; i=1,2,...,n.
\end{equation}
Then the random vector  $\textbf{q}{}^{\textbf{l}}=\left(q_{1},e^{r\left(t_{2}-t_{1}\right)}q_{1},\ldots,e^{r\left(t_{n}-t_{1}\right)}q_{1}\right)$ is comonotone. Equation \eqref{4.1.17} then reduces to
\begin{equation}\label{4.3.2}
P_{1}\geq De^{-rT}{\displaystyle \sum_{i=1}^{n}}\mbox{E}\ensuremath{\left[\left(5\left(e^{r\left(t_{i}-t_{1}\right)}q_{1}-1.3q_{0}\right)^{+}-F_{Z_{i}}^{-1}\left(F_{Z}\left(q_{0}\right)\right)\right)^{+}\right]}-K_{1},
\end{equation}
where $K_{1}$ is given in \eqref{4.1.17a} and by the definition of cdf, we have
\[
F_{Z}\left(q_{0}\right)=\textbf{P}\left[Z\leq q_{0}\right]=\textbf{P}\left[{\displaystyle \sum_{j=1}^{n}}5\left(e^{r\left(t_{j}-t_{1}\right)}q_{1}-1.3q_{0}\right)^{+}\leq q_{0}\right]
\]
\[
\Rightarrow F_{Z}\left(q_{0}\right)=\textbf{P}\left[{\displaystyle \sum_{j=1}^{n}}5\left(e^{r\left(t_{j}-t_{1}\right)}\frac{q_{1}}{q_{0}}-1.3\right)^{+}\leq 1\right].
\]

Now, as the left hand side of the inequality within the probability is an increasing function in $q_{1}/q_{0}$, we have that $Z\leq q_{0}$ if and only if $q_{1}\leq xq_{0}$, where we substitute $x$ for $q_{1}/q_{0}$ in the above probability and obtain its value by solving
\begin{equation}\label{4.3.2.1}
{\displaystyle \sum_{i=1}^{n}}\left(e^{r\left(t_{i}-t_{1}\right)}x-1.3\right)^{+}=0.2.
\end{equation}

As a result, we have
\begin{equation}\label{4.3.2.2}
F_{Z}\left(q_{0}\right)=F_{q_{1}}\left(xq_{0}\right)=F_{Z_{i}}\left(5q_{0}\left(e^{r\left(t_{i}-t_{1}\right)}x-1.3\right)^{+}\right)=F_{Z_{i}}\left(k_{i}\right)\;\;\forall i
\end{equation}
where
\begin{equation}\label{4.3.2.2a}
k_{i}=5q_{0}\left(e^{r\left(t_{i}-t_{1}\right)}x-1.3\right)^{+}.
\end{equation}
Plugging \eqref{4.3.2.2} into \eqref{4.3.2}, and noting that $Z_{i}'s$ are non-negative, we have
{
\allowdisplaybreaks
\begin{eqnarray}\label{4.3.3}
P_{1} & \geq & 5De^{-rT}{\displaystyle \sum_{i=1}^{n}}\mbox{E}\ensuremath{\left[\left(\left(e^{r\left(t_{i}-t_{1}\right)}q_{1}-1.3q_{0}\right)^{+}-\frac{1}{5}F_{Z_{i}}^{-1}\left(F_{Z_{i}}\left(k_{i}\right)\right)\right)^{+}\right]}-K_{1}
\nonumber \\
   & = &5D{\displaystyle \sum_{i=1}^{n}}e^{-r\left(T-t_{i}\right)}C\left(\frac{q_{0}}{e^{r\left(t_{i}-t_{1}\right)}}\left(1.3+\frac{1}{5q_{0}}F_{Z_{i}}^{-1}\left(F_{Z_{i}}\left(k_{i}\right)\right)\right),\; t_{1}\right)-K_{1}
\nonumber \\
& =: & \,\mbox{lb}_{1}.
\end{eqnarray}
}where $k_{i}$ is defined in \eqref{4.3.2.2a} and $q_{0} \geq 0$ and $C\left(K,t_{1}\right)$ denotes the price of a European call on the mortality index with strike K, maturity $t_{1}$ and current mortality index $q_{0}$. The function $\mbox{lb}_{1}$ provides a lower bound for the Asian-type call option in terms of European calls at each of the times such that these contracts have maturity at $t_{1}$ and the strike given by the expression $\frac{q_{0}}{e^{r\left(t_{i}-t_{1}\right)}}\left(1.3+\frac{1}{5q_{0}}F_{Z_{i}}^{-1}\left(F_{Z_{i}}\left(5q_{0}\left(e^{r\left(t_{i}-t_{1}\right)}x-1.3\right)^{+}\right)\right)\right)$ at the $i$th time point. This bound holds for any arbitrage-free market model and is a significant improvement over the trivial bound given in \eqref{4.2.5}. Invoking the put-call parity derived in section 2, the corresponding lower bound for the Swiss Re mortality bond is given as
\begin{equation}\label{4.3.8}
P\geq \left(\mbox{ lb}_{1}-G\right)^{+}=:\mbox{ SWLB}_{1},
\end{equation}
where G is defined in \eqref{2.13}. In case if the marginal cdfs $F_{Z_{i}}$ are strictly increasing, we have
\begin{equation}\label{4.3.8a}
\mbox{lb}_{1}=5D{\displaystyle \sum_{i=1}^{n}}e^{-r\left(T-t_{i}\right)}C\left(q_{0}\max\left(x,\,\frac{1.3}{e^{r\left(t_{i}-t_{1}\right)}}\right),\; t_{1}\right).
\end{equation}

\subsection{A Model-independent Lower Bound}
As the next step, we suggest that the bound $\mbox{ SWLB}_{1}$ can be improved by imposing the following additional assumption
\begin{equation}\label{4.5.1}
{\displaystyle \sum_{i=1}^{n}q_{i}\geq_{sl}\left(\sum_{i=1}^{j-1}q_{0}^{\left(1-t_{i}/t\right)}q_{t}^{t_{i}/t}+{\displaystyle \sum_{i=j}^{n}}e^{r\left(t_{i}-t\right)}q_{t}\right)}
\end{equation}
for $0\leq t\leq T$ and $j=\min\left\{ i\,:\, t_{i}\geq t\right\}$. This assumption is in the spirit of the equation 11 in \cite{prime1}. It can be shown that \eqref{4.5.1} holds good for stationary exponential L\`{e}vy models with mortality evolution $q_{t}=q_{0}e^{X_{t}}$, where $\left(X_{t}\right)_{t\geq0}$ is a L\`{e}vy process.

Clearly,
{
\allowdisplaybreaks
\begin{eqnarray}\label{4.5.5}
{\displaystyle \sum_{i=1}^{n}}5\left(\mbox{E}\left(q_{i}|q_{t}\right)-1.3q_{0}\right)^{+} & = & \sum_{i=1}^{j-1}5\left(\mbox{E}\left(q_{i}|q_{t}\right)-1.3q_{0}\right)^{+}\nonumber\\
& {} & {}+{\displaystyle \sum_{i=j}^{n}}5\left(\mbox{E}\left(q_{i}|q_{t}\right)-1.3q_{0}\right)^{+}
\nonumber\\  & \geq & \sum_{i=1}^{j-1}5q_{0}\left(\left(\frac{q_{t}}{q_{0}}\right)^{t_{i}/t}-1.3\right)^{+}\nonumber\\
& {} & {} +{\displaystyle \sum_{i=j}^{n}}5q_{0}\left(\frac{q_{t}}{q_{0}}e^{r\left(t_{i}-t\right)}-1.3\right)^{+}
\nonumber\\
& =: & S^{l_{2}}.
\end{eqnarray}
}

Evidently, $S^{l_{2}}$ is the same as $Z$ in \eqref{4.1.8} with $\Lambda$ being replaced by $q_{t}$ and thus from \eqref{4.1.11}, we have
\begin{equation}\label{4.5.6}
S \geq_{sl} S^{l_{2}}
\end{equation}

As before, let $j=\min\left\{ i\,:\, t_{i}\geq t\right\}$. Consider the components of $S^{l_{2}}$ in equation \eqref{4.5.5} and define $\textbf{Y}=\left(Y_{1},\ldots,Y_{n}\right)$, where
\[
Y_{i}=\begin{cases}
5q_{0}\left(\left(\frac{q_{t}}{q_{0}}\right)^{t_{i}/t}-1.3\right)^{+} & i<j\\
5q_{0}\left(\left(\frac{q_{t}}{q_{0}}\right)e^{r\left(t_{i}-t\right)}-1.3\right)^{+} & i\geq j
\end{cases}
\]
$i=1,2,...,n$.
Clearly, \textbf{Y} is comonotonic since its components are strictly increasing functions of a single variable $q_{t}$. So, the stop-loss transform of $S^{l_{2}}$ can be written as the sum of stop-loss transform of its components (see for example in \cite{1}), i.e.,
\begin{eqnarray}\label{4.5.7}
\mbox{E}\ensuremath{\left[\left(S^{l_{2}}-q_{0}\right)^{+}\right] & = & {\displaystyle \sum_{i=1}^{n}}\mbox{E}\ensuremath{\left[\left(Y_{i}-F_{Y_{i}}^{-1}\left(F_{S^{l_{2}}}\left(q_{0}\right)\right)\right)^{+}\right]}}\nonumber\\
& {} & {}-\left(q_{0}-F_{S^{l_{2}}}^{-1}\left(F_{S^{l_{2}}}\left(q_{0}\right)\right)\right)\left(1-F_{S^{l_{2}}}\left(q_{0}\right)\right)
\end{eqnarray}
where as before it is natural that $q_{0} \in \left(F_{S^{l_{2}}}^{-1+}\left(0\right), F_{S^{l_{2}}}^{-1}\left(1\right)\right)$ and $F_{S^{l_{2}}}\left(q_{0}\right)$ is the distribution function of $S^{l_{2}}$ evaluated at $q_{0}$ such that for an arbitrary $t$, we have:
\begin{eqnarray*}
F_{S^{l_{2}}}\left(q_{0}\right) & = & \textbf{P}\left[S^{l_{2}}\leq q_{0}\right]
\nonumber\\
& = & \textbf{P}\left[\sum_{i=1}^{j-1}\left(\left(\frac{q_{t}}{q_{0}}\right)^{t_{i}/t}-1.3\right)^{+}+{\displaystyle \sum_{i=j}^{n}}\left(\left(\frac{q_{t}}{q_{0}}\right)e^{r\left(t_{i}-t\right)}-1.3\right)^{+}\leq 0.2\right].
\nonumber
\end{eqnarray*}

Clearly, $S^{l_{2}}\leq q_{0}$ if and only if $q_{t}\leq xq_{0}$, where we substitute $x$ for $q_{t}/q_{0}$ in the above expression and obtain its value by solving:
\begin{equation}\label{4.5.9}
\sum_{i=1}^{j-1}\left(x^{t_{i}/t}-1.3\right)^{+}+{\displaystyle \sum_{i=j}^{n}}\left(xe^{r\left(t_{i}-t\right)}-1.3\right)^{+}= 0.2.
\end{equation}
As a result, we have:
\begin{equation}\label{4.5.9a}
F_{S^{l_{2}}}\left(q_{0}\right)=F_{q_{t}}\left(xq_{0}\right)=\begin{cases}
F_{Y_{i}}\left(5q_{0}\left(x^{t_{i}/t}-1.3\right)^{+}\right)=F_{Y_{i}}\left(l_{i}\right)\: & i<j\\
F_{Y_{i}}\left(5q_{0}\left(xe^{r\left(t_{i}-t\right)}-1.3\right)^{+}\right)=F_{Y_{i}}\left(m_{i}\right)\: & i\geq j
\end{cases}
\end{equation}
where
\begin{equation}\label{4.5.9b}
l_{i}=5q_{0}\left(x^{t_{i}/t}-1.3\right)^{+}; i<j
\end{equation}
and
\begin{equation}\label{4.5.9c}
m_{i}=5q_{0}\left(xe^{r\left(t_{i}-t\right)}-1.3\right)^{+}; i>j.
\end{equation}
Using this result in equation \eqref{4.5.7} and recalling the definition of the Asian-type call option given in \eqref{4.1.1} along with the stop-loss order relationship between $S$ and $S^{l_{2}}$ as given by equation \eqref{4.5.6} and noting that $Y_{i}'s$ are non-negative, we obtain,
{
\allowdisplaybreaks
\begin{eqnarray}\label{4.5.10}
P_{1} & \geq & De^{-rT}\left({\displaystyle \sum_{i=1}^{n}}\mbox{E}\ensuremath{\left[\left(Y_{i}-F_{Y_{i}}^{-1}\left(F_{S^{l_{2}}}\left(q_{0}\right)\right)\right)^{+}\right]}\right)-K_{2}
\nonumber\\
& = & 5De^{-rT}\Bigg(\sum_{i=1}^{j-1}q_{0}^{1-t_{i}/t}\mbox{E}\left[\left(q_{t}^{t_{i}/t}-q_{0}^{t_{i}/t}\left(1.3+\frac{1}{5q_{0}}F_{Y_{i}}^{-1}\left(F_{Y_{i}}\left(l_{i}\right)\right)\right)\right)^{+}\right] \nonumber\\
& {} {} & {} \;\;\;\;\;\;\;\;+{\displaystyle \sum_{i=j}^{n}}e^{rt_{i}}C\left(\frac{q_{0}}{e^{r\left(t_{i}-t\right)}}\left(1.3+\frac{1}{5q_{0}}F_{Y_{i}}^{-1}\left(F_{Y_{i}}\left(m_{i}\right)\right)\right),\; t\right)\Bigg)-K_{2} \nonumber\\
& =: & \,\mbox{lb}_{t}^{\left(2\right)}
\end{eqnarray}
}
where $l_{i}$ and $m_{i}$ are defined in \eqref{4.5.9b} and \eqref{4.5.9c} respectively and
\begin{equation}\label{4.5.10a}
K_{2}=De^{-rT}\left(q_{0}-F_{S^{l_{2}}}^{-1}\left(F_{S^{l_{2}}}\left(q_{0}\right)\right)\right)\left(1-F_{S^{l_{2}}}\left(q_{0}\right)\right)
\end{equation}

In fact, $\mbox{lb}_{t}^{\left(2\right)}$ is a lower bound for all $t$ and so it can be maximized with respect to $t$ to yield the optimal lower bound as given below:
\begin{equation}\label{4.5.19}
P_{1}\geq\max_{0\leq t\leq T}\mbox{lb}_{t}^{\left(2\right)}.
\end{equation}

On choosing $t=t_{1}$ implies $j=1$ and so equation \eqref{4.5.9} reduces to \eqref{4.3.2.1} and we obtain
\begin{equation}
\mbox{lb}_{1}^{\left(2\right)}=\mbox{lb}_{1}.
\end{equation}
As a result we have
\[
\max_{0\leq t\leq T}\mbox{lb}_{t}^{\left(2\right)} \geq \mbox{lb}_{1}.
\]
Clearly, once again, as in the previous sections, we have
\begin{equation}\label{4.5.20}
P\geq \left(\mbox{lb}_{t}^{\left(2\right)}-G\right)^{+}=:\mbox{SWLB}_{t}^{\left(2\right)},
\end{equation}
where G is defined in \eqref{2.13}. In case if the marginal cdfs $F_{Y_{i}}$ are strictly increasing, we have
\begin{eqnarray}\label{4.5.20a}
\mbox{lb}_{t}^{\left(2\right)} & = & 5De^{-rT}\Bigg(\sum_{i=1}^{j-1}q_{0}^{1-t_{i}/t}\mbox{E}\left[\left(q_{t}^{t_{i}/t}-q_{0}^{t_{i}/t}\max\left(x^{t_{i}/t},\,1.3\right)\right)^{+}\right] \nonumber\\
& {} {} & {} \;\;\;\;\;\;\;\;+{\displaystyle \sum_{i=j}^{n}}e^{rt_{i}}C\left(q_{0}\max\left(x,\,\frac{1.3}{e^{r\left(t_{i}-t\right)}}\right),\; t\right)\Bigg).
\end{eqnarray}

We now move on to the derivation of an upper bound for the price of Swiss Re bond in the next section.

\section{Upper Bounds for the Swiss Re Bond}
We derive a couple of upper bounds for the Swiss Re bond.
\subsection{A First Upper Bound}
This section will focus on finding an upper bound for the bond in question by using comonotonicity theory in a manner similar to \cite{Upper} and \cite{Dhaene2}. Define the comonotonic counterpart of $\textbf{q}{}=\left(q_{1},...,q_{n}\right)$ as $\textbf{q}{}^{\textbf{u}}=\left(F_{S_{1}}^{-1}\left(U\right),...,F_{S_{n}}^{-1}\left(U\right)\right)$
where $U \sim U\left(0,1\right)$. Further define
\begin{equation}\label{4.1.1a}
S^{c}={\displaystyle \sum_{i=1}^{n}}F_{S_{i}}^{-1}\left(U\right)={\displaystyle \sum_{i=1}^{n}}S_{i}^{c}.
\end{equation}
Clearly,
\begin{equation}\label{4.1.2a}
S \leq_{cx} S^{c}
\end{equation}
where $cx$ denotes convex ordering (see for example in \cite{1}). In other words,
\begin{equation}\label{4.1.3a}
\mbox{E}\ensuremath{\left[\left({\displaystyle \sum_{i=1}^{n}}S_{i}-q_{0}\right)^{+}\right]} \leq \mbox{E}\ensuremath{\left[\left({\displaystyle \sum_{i=1}^{n}}S_{i}^{c}-q_{0}\right)^{+}\right]}
\end{equation}
and we have
\begin{eqnarray}\label{4.1.3ab}
\mbox{E}\ensuremath{\left[\left({\displaystyle \sum_{i=1}^{n}}S_{i}^{c}-q_{0}\right)^{+}\right]} & = &  {\displaystyle \sum_{i=1}^{n}}\mbox{E}\ensuremath{\left[\left(S_{i}-F_{S_{i}}^{-1}\left(F_{S^{c}}\left(q_{0}\right)\right)\right)^{+}\right]}\nonumber\\
& {} & -\left(q_{0}-F_{S^{c}}^{-1}\left(F_{S^{c}}\left(q_{0}\right)\right)\right)\left(1-F_{S^{c}}\left(q_{0}\right)\right)
\end{eqnarray}
where it is understood that $q_{0} \in \left(F_{S^{c}}^{-1+}\left(0\right), F_{S^{c}}^{-1}\left(1\right)\right)$. As a result, an upper bound for the call counterpart of the Swiss Re bond is given as
{
\allowdisplaybreaks
\begin{eqnarray}\label{4.1.4a}
P_{1} & \leq & De^{-rT}{\displaystyle \sum_{i=1}^{n}}\mbox{E}\ensuremath{\left[\ensuremath{\left(S_{i}-F_{S_{i}}^{-1}\left(F_{S^{c}}\left(q_{0}\right)\right)\right)^{+}}\right]}-K_{3}\nonumber\\
    & = & 5De^{-rT}{\displaystyle \sum_{i=1}^{n}e^{rt_{i}}C\left(1.3q_{0}+\frac{F_{S_{i}}^{-1}\left(F_{S^{c}}\left(q_{0}\right)\right)}{5},t_{i}\right)}-K_{3},
\end{eqnarray}
}
where
\begin{equation}\label{4.1.4ab}
K_{3}=De^{-rT}\left(q_{0}-F_{S^{c}}^{-1}\left(F_{S^{c}}\left(q_{0}\right)\right)\right)\left(1-F_{S^{c}}\left(q_{0}\right)\right).
\end{equation}
As a result we can write the upper bound given above as
\begin{equation}\label{4.1.5a}
P_{1} \leq 5De^{-rT}{\displaystyle \sum_{i=1}^{n}e^{rt_{i}}C\left(1.3q_{0}+\frac{F_{S_{i}}^{-1}\left(x\right)}{5},t_{i}\right)}-K_{3},
\end{equation}
where $x\in\left(0,1\right)$ is the solution of the equation
\begin{equation}\label{4.1.6a}
{\displaystyle \sum_{i=1}^{n}F_{S_{i}}^{-1}\left(x\right)=q_{0}}.
\end{equation}
We now seek to express the inverse distribution function of $S_{i}$ in terms of that of $q_{i}$. Let
\begin{equation}\label{4.1.7a}
y_{i} = F_{S_{i}}^{-1}\left(x\right);\;y_{i} \geq 0
\end{equation}
\begin{eqnarray}\label{4.1.8a}
\Rightarrow x & = & F_{S_{i}}\left(y_{i}\right) \nonumber\\
& = & P\left[5\left(q_{i}-1.3q_{0}\right)^{+} \leq y_{i}\right] \nonumber\\
& = & F_{q_{i}}\left(1.3q_{0}+\frac{y_{i}}{5}\right).
\end{eqnarray}
\begin{equation}\label{4.1.9a}
\Rightarrow y_{i} = 5\left(F_{q_{i}}^{-1}\left(x\right)-1.3q_{0}\right).
\end{equation}

From equations \eqref{4.1.5a}, \eqref{4.1.7a} and \eqref{4.1.9a}, we conclude that the upper bound is given as
\begin{equation}\label{4.1.10a}
P_{1}\leq 5De^{-rT}{\displaystyle \sum_{i=1}^{n}e^{rt_{i}}C\left(F_{q_{i}}^{-1}\left(x\right),t_{i}\right)}-K_{3}=: \mbox{ub}_{1}.
\end{equation}
where using equations \eqref{4.1.6a} and \eqref{4.1.9a}, we see that $x$ solves the following equation
\begin{equation}\label{4.1.11a}
{\displaystyle \sum_{i=1}^{n}F_{q_{i}}^{-1}\left(x\right)=\frac{q_{0}}{5}\left(1+6.5n\right)}.
\end{equation}

As in the case of lower bounds, invoking the put-call parity of section 2, we have for the Swiss Re bond
\begin{equation}\label{5.10}
P\leq \left(\mbox{ub}_{1}-G\right)^{+}=:\mbox{SWUB}_{1}.
\end{equation}
where G is defined in \eqref{2.13}. In case if the marginal cdfs $F_{S_{i}}$ are strictly increasing, we have
\begin{equation}\label{5.10ab}
\mbox{ub}_{1} = 5De^{-rT}{\displaystyle \sum_{i=1}^{n}e^{rt_{i}}C\left(F_{q_{i}}^{-1}\left(x\right),t_{i}\right)}
\end{equation}
where $x$ solves the equation \eqref{4.1.11a}.

\subsection{An Improved Upper Bound by conditioning}
We now seek to obtain a sharper upper bound for the Swiss Re bond. This is possible if we assume that some additional information is available concerning the stochastic nature of $\left(q_{1},q_{2},...,q_{n}\right)$. That is, if we can find a random variable $\Lambda$, with a known distribution, such that the individual conditional distributions of $q_{i}$ given the event $\Lambda=\lambda$ are known for all $i$ and all possible values of $\lambda$. Such an approach can be found in \cite{Upper}, \cite{1}, \cite{2}, \cite{Laeven} and \cite{Goovaerts21}.

Define
\begin{equation}\label{4.22}
S^{u}={\displaystyle \sum_{i=1}^{n}}F_{S_{i}|\Lambda}^{-1}\left(U\right)={\displaystyle \sum_{i=1}^{n}}S_{i}^{u},
\end{equation}
where $U \sim U\left(0,1\right)$. Then we have
\begin{equation}\label{4.21}
S \leq_{cx} S^{u} \leq_{cx} S^{c},
\end{equation}
Now let $\textbf{q}{}^{\textbf{u}}=\left(S_{1}^{u},...,S_{n}^{u}\right)$. Since $\left(F_{S_{1}|\Lambda=\lambda}^{-1},...,F_{S_{n}|\Lambda=\lambda}^{-1}\right)$ is comonotonic, we have,
\begin{equation}\label{4.23}
F_{S^{u}|\Lambda=\lambda}^{-1}\left(p\right)={\displaystyle \sum_{i=1}^{n}}F_{S_{i}|\Lambda=\lambda}^{-1}\left(p\right),\;p \in \left(0,1\right).
\end{equation}
It follows that, in this case
\begin{equation}\label{4.24}
{\displaystyle \sum_{i=1}^{n}}F_{S_{i}|\Lambda=\lambda}^{-1}\left(F_{S^{u}|\Lambda=\lambda}\left(q_{0}\right)\right)=q_{0}.
\end{equation}
and so we have
\begin{eqnarray}\label{4.25}
f\left(\lambda\right) & = & \mbox{E}\ensuremath{\left[\left({\displaystyle \sum_{i=1}^{n}}S_{i}^{u}-q_{0}\right)^{+}\middle|\Lambda=\lambda\right]} \nonumber\\
& = & {\displaystyle \sum_{i=1}^{n}}\mbox{E}\ensuremath{\left[\left(S_{i}-F_{S_{i}|\Lambda=\lambda}^{-1}\left(F_{S^{u}|\Lambda=\lambda}\left(q_{0}\right)\right)\right)^{+}\middle|\Lambda=\lambda\right]}\nonumber \\
& & -\left(q_{0}-F_{S^{u}|\Lambda=\lambda}^{-1}\left(F_{S^{u}|\Lambda=\lambda}\left(q_{0}\right)\right)\right)\left(1-F_{S^{u}|\Lambda=\lambda}\left(q_{0}\right)\right),
\end{eqnarray}
where it is clear that $q_{0} \in \left(F_{S^{u}|\Lambda=\lambda}^{-1+}\left(0\right), F_{S^{u}|\Lambda=\lambda}^{-1}\left(1\right)\right)$. By applying the tower property and using the convex order relationship given by \eqref{4.21}, we obtain an upper bound for the call counterpart of the Swiss Re bond, i.e.,
{
\allowdisplaybreaks
\begin{eqnarray}\label{4.26}
P_{1} & \leq & De^{-rT}\mbox{E}\ensuremath{\left[\left(S^{u}-q_{0}\right)^{+}\right]}\nonumber \\
   & = & De^{-rT}\mbox{E}\ensuremath{\left[f\left(\lambda\right)\right]}\nonumber \\
   & = & De^{-rT}{\displaystyle \sum_{i=1}^{n}}{\displaystyle \int_{-\infty}^{\infty}}\mbox{E}\ensuremath{\left[\ensuremath{\left(S_{i}-F_{S_{i}|\Lambda=\lambda}^{-1}\left(F_{S^{u}|\Lambda=\lambda}\left(q_{0}\right)\right)\right)^{+}\middle|\Lambda=\lambda}\right]}dF_{\Lambda}\left(\lambda\right)-K_{4} \nonumber\\
& = & 5De^{-rT}I_{1}-K_{4}\nonumber
\end{eqnarray}
}
where
\begin{equation}
I_{1}={\displaystyle \sum_{i=1}^{n}}{\displaystyle \int_{-\infty}^{\infty}}\mbox{E}\ensuremath{\left[\left(q_{i}-\left(1.3q_{0}+\frac{F_{S_{i}|\Lambda=\lambda}^{-1}\left(F_{S^{u}|\Lambda=\lambda}\left(q_{0}\right)\right)}{5}\right)\right)^{+}\middle|\Lambda=\lambda\right]}dF_{\Lambda}\left(\lambda\right) \nonumber
\end{equation}
and
\begin{equation}\label{4.26a}
K_{4}=\left(q_{0}-F_{S^{u}|\Lambda=\lambda}^{-1}\left(F_{S^{u}|\Lambda=\lambda}\left(q_{0}\right)\right)\right)\left(1-F_{S^{u}|\Lambda=\lambda}\left(q_{0}\right)\right).
\end{equation}
Given the event $\Lambda=\lambda$, let $x$ be the solution to the following equation
\begin{equation}\label{4.272}
{\displaystyle \sum_{i=1}^{n}}F_{S_{i}|\Lambda=\lambda}^{-1}\left(x\right)=q_{0}.
\end{equation}
Further, we see from equation \eqref{4.24}, that $x=F_{S^{u}|\Lambda=\lambda}\left(q_{0}\right)$. It therefore follows, as a result of equation 93 of \cite{1} that an upper bound for the call counterpart of the Swiss Re bond is given as
\begin{equation}\label{4.28}
P_{1} \leq 5De^{-rT}{\displaystyle \sum_{i=1}^{n}}{\displaystyle \int_{-\infty}^{\infty}}\mbox{E}\ensuremath{\left[\left(q_{i}-\left(1.3q_{0}+\frac{F_{S_{i}|\Lambda=\lambda}^{-1}\left(x\right)}{5}\right)\right)^{+}\middle|\Lambda=\lambda\right]}dF_{\Lambda}\left(\lambda\right)-K_{4},
\end{equation}
where $x$ is obtained by solving \eqref{4.272}. Moreover, it is straightforward to write
\begin{equation}\label{4.28a}
F_{S_{i}|\Lambda=\lambda}^{-1}\left(x\right) = 5\left(F_{q_{i}|\Lambda=\lambda}^{-1}\left(x\right)-1.3q_{0}\right).
\end{equation}
As a result, the upper bound can be rewritten as
\begin{equation}\label{4.28b}
P_{1} \leq 5De^{-rT}{\displaystyle \sum_{i=1}^{n}}{\displaystyle \int_{-\infty}^{\infty}}\mbox{E}\ensuremath{\left[\left(q_{i}-F_{q_{i}|\Lambda=\lambda}^{-1}\left(x\right)\right)^{+}\middle|\Lambda=\lambda\right]}dF_{\Lambda}\left(\lambda\right)-K_{4}=:\mbox{ub}_{t}^{\left(1\right)},
\end{equation}
where $x \in \left(0,1\right)$ can be obtained by solving the equation
\begin{equation}\label{4.28c}
{\displaystyle \sum_{i=1}^{n}}F_{q_{i}|\Lambda=\lambda}^{-1}\left(x\right)=\frac{q_{0}}{5}\left(1+6.5n\right).
\end{equation}

Since this is is an upper bound for all $t$, it follows that we can find the optimal upper bound by minimising equation \eqref{4.28b} over $t\in \left[0,T\right]$. As before, invoking the put-call parity of section 2, we have for the Swiss Re bond
\begin{equation}\label{4.29}
P \leq \left(\mbox{ub}_{t}^{\left(1\right)}-G\right)^{+}=:\mbox{SWUB}_{t}^{\left(1\right)},
\end{equation}
where G is defined in \eqref{2.13}. As remarked earlier, this bound improves upon the unconditional bound given by \eqref{5.10}. In case if the marginal cdfs $F_{S_{i}|\Lambda}$ are strictly increasing, one can put $K_{4}=0$ in \eqref{4.28b} to obtain the upper bound.

\section{Examples}
We now derive lower and upper bounds by choosing specific models for the mortality index.
\subsection{Black-Scholes Model}
Let us consider the case where the mortality evolution process $\left\{ q_{t}\right\}_{t\geq0}$ follows the Black-Scholes model (c.f. \cite{Black}) which we write as $q_{t}=e^{U_{t}}$, where $\left\{ U_{t}\right\} _{t\geq0}$ is defined as:
\begin{equation}\label{4.6.1}
U_{t}=\log_{e}\left(q_{0}\right)+\left(r-\frac{\sigma^{2}}{2}\right)t+\sigma W_{t}^{*},
\end{equation}
where $\left\{ W_{t}^{*}\right\} _{t\geq0}$ denotes a standard Brownian motion so that $W_{t}^{*}\sim N\left(0,t\right)$. As a result
\begin{equation}\label{4.6.2}
U_{t}\sim N\left(\log_{e}q_{0}+\left(r-\frac{\sigma^{2}}{2}\right)t,\,\sigma^{2}t\right).
\end{equation}
We now derive lower and upper bounds for this model on the lines of $\mbox{SWLB}_{t}^{\left(2\right)}$ and $\mbox{SWUB}_{t}^{\left(1\right)}$ respectively.

\subsubsection{The Lower Bound $\mbox{SWLB}_{t}^{\left(BS\right)}$}

We know that if $\left(X,\, Y\right)\sim\mbox{BVN}\left(\mu_{X},\mu_{Y},\sigma_{X}^{2},\sigma_{Y}^{2},\rho\right)$ where $BVN$ stands for bivariate normal distribution, the conditional distribution of the lognormal random variable $e^{X}$, given the event $e^{Y}=y$ is given as
\begin{equation}\label{4.6.3}
F_{e^{X}|e^{Y}=y}\left(x\right)=\Phi\left(\frac{\log_{e}x-\left(\mu_{X}+\rho\frac{\sigma_{X}}{\sigma_{Y}}\left(\log_{e}y-\mu_{Y}\right)\right)}{\sigma_{X}\sqrt{1-\rho^{2}}}\right).
\end{equation}
where $\Phi$ denotes the c.d.f. of standard normal distribution. Given the time points $t_{i}$, $t$ for each $i$, let $\rho$ be the correlation between $U_{t_{i}}$ and $U_{t}$. Then, from \eqref{4.6.2}, it is evident that: $\left(U_{t_{i}},U_{t}\right)\sim\mbox{BVN}\left(\mu_{U_{t_{i}}},\mu_{U_{t}},\sigma_{U_{t_{i}}}^{2},\sigma_{U_{t}}^{2},\rho\right)$, where the same equation specifies $\mu_{U_{t_{i}}},\mu_{U_{t}},\sigma_{U_{t_{i}}}^{2}$ and $\sigma_{U_{t}}^{2}$. Also as $q_{t}=e^{U_{t}}$, we have from equation \eqref{4.6.3} that the distribution function of $q_{i}$ conditional on the event $q_{t}=s_{t}$ is given as
\[
F_{q_{i}|q_{t}=s_{t}}\left(x\right)=\Phi\left(a\left(x\right)\right)
\]
where $a\left(x\right)$ is given by
\begin{equation}\label{4.6.4}
a\left(x\right)=\frac{\log_{e}x-\left(\log\left(q_{0}\left(\frac{s_{t}}{q_{0}}\right)^{\rho\sqrt{\frac{t_{i}}{t}}}\right)+\left(r-\frac{\sigma^{2}}{2}\right)\left(t_{i}-\rho\sqrt{t_{i}t}\right)\right)}{\sigma\sqrt{t_{i}\left(1-\rho^{2}\right)}}.
\end{equation}
As the differentiation of c.d.f. yields the p.d.f., therefore the conditional density function of $q_{i}$ given $q_{t}=s_{t}$ satisfies the following equation:
\begin{equation}\label{4.6.5}
f_{q_{i}|q_{t}=s_{t}}\left(x\right)=\frac{1}{x\sigma\sqrt{t_{i}\left(1-\rho^{2}\right)}}\phi\left(a\left(x\right)\right),
\end{equation}
where $\phi$ denotes the p.d.f. of standard normal distribution. 

Under the assumption that the mortality evolution process $\left\{ q_{t}\right\} _{t\geq0}$ is defined as $q_{t}=e^{U_{t}}$ where $U_{t}$ is given in equation \eqref{4.6.1}, the conditional expectation of $q_{i}$ given $q_{t}$ is given by the expression
\begin{equation}\label{4.6.6}
\mbox{E}\left(q_{i}|q_{t}\right)=\begin{cases}
q_{0}\left(\frac{q_{t}}{q_{0}}\right)^{\frac{t_{i}}{t}}e^{\frac{\sigma^{2}t_{i}}{2t}\left(t-t_{i}\right)}\;\;\; & t_{i}<t,\\
q_{t}e^{r\left(t_{i}-t\right)} & t_{i}\geq t.
\end{cases}
\end{equation}
We utilize this expression to obtain a lower bound for Asian call option under the Black-Scholes setting. Define: $S^{l_{3}}=\sum_{i=1}^{n}Y_{i}$, where exploiting \eqref{4.6.6}, under the Black-Scholes case, $Y_{i}$, $i=1,2,...,n$ are given by
\[
Y_{i}=\begin{cases}
5q_{0}\left(\left(\frac{q_{t}}{q_{0}}\right)^{t_{i}/t}e^{\frac{\sigma^{2}t_{i}}{2t}\left(t-t_{i}\right)}-1.3\right)^{+}\;\; & i<j\\
5q_{0}\left(\left(\frac{q_{t}}{q_{0}}\right)e^{r\left(t_{i}-t\right)}-1.3\right)^{+} & i\geq j
\end{cases}
\]
Evidently, $\textbf{Y}=\left(Y_{1},\ldots,Y_{n}\right)$ is comonotonic and so we have
\begin{equation}\label{4.6.8}
\mbox{E}\ensuremath{\left[\left(S^{l_{3}}-q_{0}\right)^{+}\right]={\displaystyle \sum_{i=1}^{n}}\mbox{E}\ensuremath{\left[\left(Y_{i}-F_{Y_{i}}^{-1}\left(F_{S^{l_{3}}}\left(q_{0}\right)\right)\right)^{+}\right]}},
\end{equation}
where $F_{S^{l_{3}}}\left(q_{0}\right)$ is the distribution function of $S^{l_{3}}$ evaluated at $q_{0}$. For an arbitrary t, we have
\begin{eqnarray}\label{4.6.9}
F_{S^{l_{3}}}\left(q_{0}\right) & = & \textbf{P}\left[S^{l_{3}}\leq q_{0}\right]
\nonumber\\
& = & \textbf{P}\Bigg[\sum_{i=1}^{j-1}5q_{0}\left(\left(\frac{q_{t}}{q_{0}}\right)^{t_{i}/t}e^{\frac{\sigma^{2}t_{i}}{2t}\left(t-t_{i}\right)}-1.3\right)^{+}\nonumber\\
& & {}+{\displaystyle \sum_{i=j}^{n}}5q_{0}\left(\left(\frac{q_{t}}{q_{0}}\right)e^{r\left(t_{i}-t\right)}-1.3\right)^{+}\leq q_{0}\Bigg]
\nonumber\\
& = & \textbf{P}\Bigg[\sum_{i=1}^{j-1}\left(\left(\frac{q_{t}}{q_{0}}\right)^{t_{i}/t}e^{\frac{\sigma^{2}t_{i}}{2t}\left(t-t_{i}\right)}-1.3\right)^{+}\nonumber\\
& & {} +{\displaystyle \sum_{i=j}^{n}}\left(\left(\frac{q_{t}}{q_{0}}\right)e^{r\left(t_{i}-t\right)}-1.3\right)^{+}\leq 0.2\Bigg].
\end{eqnarray}
As in the previous section, we substitute $x$ for $q_{t}/q_{0}$ and solve for $x$, using the equation:
\begin{equation}\label{4.6.10}
\sum_{i=1}^{j-1}\left(x^{t_{i}/t}e^{\frac{\sigma^{2}t_{i}}{2t}\left(t-t_{i}\right)}-1.3\right)^{+}+{\displaystyle \sum_{i=j}^{n}}\left(xe^{r\left(t_{i}-t\right)}-1.3\right)^{+}=0.2.
\end{equation}
This is indeed straight forward, noting that the left hand side of this equation is strictly increasing in $x$. This yields:
\[
F_{S^{l_{3}}}\left(q_{0}\right)=F_{q_{t}}\left(xq_{0}\right)=\begin{cases}
F_{Y_{i}}\left(5q_{0}\left(x^{t_{i}/t}e^{\frac{\sigma^{2}t_{i}}{2t}\left(t-t_{i}\right)}-1.3\right)^{+}\right)\: & i<j,\\
F_{Y_{i}}\left(5q_{0}\left(xe^{r\left(t_{i}-t\right)}-1.3\right)^{+}\right)\: & i\geq j.
\end{cases}
\]
Substituting this in equation \eqref{4.6.8}, recalling the stop-loss order relationship between $S$ and $S^{l_{2}}$ as given by equation \eqref{4.5.6}, applying it for $S^{l_{3}}$, splitting the terms and multiplying by the averaged discount factor as done in the last section and noting that the marginal cdfs $F_{Y_{i}}$ are strictly increasing, we obtain
{
\allowdisplaybreaks
\begin{eqnarray}\label{4.6.11}
P_{1} & \geq & De^{-rT}\left({\displaystyle \sum_{i=1}^{n}}\mbox{E}\ensuremath{\left[\left(Y_{i}-F_{Y_{i}}^{-1}\left(F_{S^{l_{3}}}\left(q_{0}\right)\right)\right)^{+}\right]}\right)
\nonumber\\
& = & 5De^{-rT}\Bigg(\sum_{i=1}^{j-1}q_{0}^{1-t_{i}/t}\mbox{E}\left[\left(q_{t}^{t_{i}/t}e^{\frac{\sigma^{2}t_{i}}{2t}\left(t-t_{i}\right)}-q_{0}^{t_{i}/t}\max\left(x^{t_{i}/t}e^{\frac{\sigma^{2}t_{i}}{2t}\left(t-t_{i}\right)},\,1.3\right)\right)^{+}\right] \nonumber\\
& {} {} & {} \;\;\;\;\;\;\;\;+{\displaystyle \sum_{i=j}^{n}}e^{rt_{i}}C\left(q_{0}\max\left(x,\,\frac{1.3}{e^{r\left(t_{i}-t\right)}}\right),\, t\right)\Bigg).
\end{eqnarray}
}
We denote the term within the first summation as $E_{1}$ and its value is given below.
\begin{equation}\label{4.6.12}
\mbox{E}_{1}=5q_{0}\left(e^{rt_{i}}\Phi\left(d_{1ai}\right)-\max\left(x^{t_{i}/t}e^{\frac{\sigma^{2}t_{i}}{2t}\left(t-t_{i}\right)},\,1.3\right)\Phi\left(d_{2ai}\right)\right),
\end{equation}
where $d_{2ai}$ and $d_{1ai}$ are given respectively as
\begin{equation}\label{4.6.13}
d_{2ai}=\frac{-\log_{e}\left(\frac{da_{i}}{q_{0}}\right)+\left(r-\frac{\sigma^{2}}{2}\right)t}{\sigma\sqrt{t}}
\end{equation}
\begin{equation}\label{4.6.14}
d_{1ai}=d_{2ai}+\sigma\frac{t_{i}}{\sqrt{t}}
\end{equation}
and $da_{i}$ is given as
\begin{equation}\label{4.6.15}
da_{i}=q_{0}\left(\max\left(x^{t_{i}/t},\,\frac{1.3}{e^{\frac{\sigma^{2}t_{i}}{2t}\left(t-t_{i}\right)}}\right)\right)^{t/t_{i}}.
\end{equation}
Inserting \eqref{4.6.12} in \eqref{4.6.11}, we achieve the lower bound $\mbox{lb}_{t}^{\left(BS\right)}$ as follows
\begin{eqnarray}\label{4.6.16}
P_{1} & \geq & 5De^{-rT}\Bigg(\sum_{i=1}^{j-1}q_{0}\left(e^{rt_{i}}\Phi\left(d_{1ai}\right)-\max\left(x^{t_{i}/t}e^{\frac{\sigma^{2}t_{i}}{2t}\left(t-t_{i}\right)},\,1.3\right)\Phi\left(d_{2ai}\right)\right) \nonumber\\
& & {} \;\;\;\;\;\;\;\;\;+{\displaystyle \sum_{i=j}^{n}}e^{rt_{i}}C\left(q_{0}\max\left(x,\,\frac{1.3}{e^{r\left(t_{i}-t\right)}}\right),\, t\right)\Bigg) \nonumber\\
& =: & \,\mbox{lb}_{t}^{\left(BS\right)}.
\end{eqnarray}
The bound $\mbox{lb}_{t}^{\left(BS\right)}$ can undergo treatment similar to $\mbox{lb}_{t}^{\left(2\right)}$ in sense of maximization with respect to $t$ yielding
\begin{equation}\label{4.6.17}
P_{1}\geq\max_{0\leq t\leq T}\mbox{lb}_{t}^{\left(BS\right)}.
\end{equation}
An interesting comment in the passing is that as we calculate $\mbox{E}\left[q_{i}|q_{t}\right]$ explicitly, rather than finding a lower bound for it, clearly $\mbox{lb}_{t}^{\left(BS\right)}$ improves on $\mbox{lb}_{t}^{\left(2\right)}$ in the case where $\left\{ q_{t}\right\}$ follows the Black-Scholes model. Again, as before, exploiting the put-call parity,
\begin{equation}\label{4.6.18}
P\geq \left(\mbox{lb}_{t}^{\left(BS\right)}-G\right)^{+}=:\mbox{SWLB}_{t}^{\left(BS\right)},
\end{equation}
where G is defined in \eqref{2.13}.

\subsubsection{The Upper Bound $\mbox{SWUB}_{t}^{\left(BS\right)}$}
In section 4.2, we have shown that the upper bound $\mbox{SWUB}_{1}$ can be improved by assuming that there exists a random variable $\Lambda$ such that $\text{Cov}\left(X_{i}, \Lambda\right) \neq 0\;\forall i$. Suppose this assumption is true here and the mortality index $\left\{ q_{t}\right\}_{t\geq0}$ depends on an underlying standard Brownian motion $\{ W_{t}\}_{t \in \left[0,T\right]}$. Then, from equation \eqref{4.28b} noting that the marginal cdfs $F_{q_{i}|W_t=w}$ are strictly increasing so that $K_{4}=0$, we see that an upper bound for the call counterpart of the Swiss Re bond is given as
\begin{equation}\label{5.2.1}
P_{1} \leq 5De^{-rT}{\displaystyle \sum_{i=1}^{n}}{\displaystyle \int_{-\infty}^{\infty}}\mbox{E}\ensuremath{\left[\ensuremath{\left(q_{i}-F_{q_{i}|W_t=w}^{-1}\left(x\right)\right)^{+}\middle|W_t=w}\right]}d\Phi\left(\frac{w}{\sqrt{t}}\right),
\end{equation}
where using \eqref{4.28c}, we see that $x$ is obtained by solving the following equation
\begin{equation}\label{5.2.2}
{\displaystyle \sum_{i=1}^{n}}F_{q_{i}|W_t=w}^{-1}\left(x\right)=\frac{q_{0}}{5}\left(1+6.5n\right).
\end{equation}

An explicit formula for the conditional inverse distribution function of $q_{i}$ given the event $W_t=w$, is provided by the following result.
\begin{prop}\label{prop2}
Under the assumptions of the Black-Scholes model, conditional on the event $W_t=w$, the conditional distribution function of $q_{i}$ is given by
\begin{equation}\label{5.2.3}
F_{q_{i}|W_t=w}^{-1}=\begin{cases}
q_{0}e^{\left(r-\frac{\sigma^{2}}{2}\right)t_{i}+\sigma\frac{t_{i}}{t}w+\sigma\sqrt{\frac{t_{i}}{t}\left(t-t_{i}\right)}\Phi^{-1}\left(x\right)}\;\;\; & i<j,\\
q_{0}e^{\left(r-\frac{\sigma^{2}}{2}\right)t_{i}+\sigma w+\sigma\sqrt{\left(t_{i}-t\right)}\Phi^{-1}\left(x\right)} & i\geq j.
\end{cases}
\end{equation}
where $j = min\{i: t_{i} \geq t\}$.
\end{prop}
\begin{proof}
Let us set $X = \sigma W_{t_{i}}$, $Y=W_{t}$ and $y=e^{w}$ in \eqref{4.6.3}. Then we obtain the following expression for the conditional distribution function of $e^{\sigma W_{t_{i}}}$ given the event $W_{t}=w$.
\begin{equation}\label{5.2.4}
F_{e^{\sigma W_{t_{i}}}|W_{t}=w}\left(s\right)=\Phi\left(\frac{\log_{e}s-\rho\sigma\sqrt{\frac{t_{i}}{t}}w}{\sigma\sqrt{t_{i}\left(1-\rho^{2}\right)}}\right).
\end{equation}
It then follows that $F_{e^{\sigma W_{t_{i}}}|W_{t}=w}\left(s\right)=x$ if and only if
\[s=F^{-1}_{e^{\sigma W_{t_{i}}}|W_{t}=w}\left(x\right)=e^{\rho\sigma\sqrt{\frac{t_{i}}{t}}w+\sigma\sqrt{t_{i}\left(1-\rho^{2}\right)}\Phi^{-1}\left(x\right)}\]
We can then obtain equation \eqref{5.2.3} by noting that $\rho=\sqrt{\left(t_{i}\wedge t\right)\left(t_{i}\vee t\right)}$ and the following expression for the inverse conditional distribution function of $q_{i}$ given $W_{t}=w$.
\[
F_{q_{i}|W_t=w}^{-1}=q_{0}e^{\left(r-\frac{\sigma^{2}}{2}\right)t_{i}}F^{-1}_{e^{\sigma W_{t_{i}}}|W_{t}=w}
\]
This completes the proof.
\end{proof}

It is of note that $F_{q_{i}|W_t=w}^{-1}$ is continuous when $t=t_{i}$ (that is if, for some $i$, we have $i=j$). From equation \eqref{5.2.2}, we then wish to solve the following for $x$.
\begin{equation}\label{5.2.5}
\sum_{i=1}^{j-1}e^{\left(r-\frac{\sigma^{2}}{2}\right)t_{i}+\sigma\frac{t_{i}}{t}w+\sigma\sqrt{\frac{t_{i}}{t}\left(t-t_{i}\right)}\Phi^{-1}\left(x\right)}+{\displaystyle \sum_{i=j}^{n}}e^{\left(r-\frac{\sigma^{2}}{2}\right)t_{i}+\sigma w+\sigma\sqrt{\left(t_{i}-t\right)}\Phi^{-1}\left(x\right)}=0.2+1.3n.
\end{equation}
As a result, using equation\eqref{5.2.1}, the improved upper bound for the call counterpart of the Swiss Re bond in the Black-Scholes case is given by the following set of equations
\begin{eqnarray}\label{5.2.5a}
P_{1} & \leq & 5Ce^{-rT}{\displaystyle \int_{-\infty}^{\infty}}\Bigg(\sum_{i=1}^{n}e^{\left(r-\frac{\sigma^{2}\left(t_{i}\wedge t\right)^{2}}{2t_{i}t}\right)t_{i}+\sigma\frac{t_{i}\wedge t}{t}w}\Phi\left(c_{1}^{\left(i\right)}\right)-\left(0.2+1.3n\right)\left(1-x\right)\Bigg)d\Phi\left(\frac{w}{\sqrt{t}}\right) \nonumber\\
& =: & \,\mbox{ub}_{t}^{\left(BS\right)},
\end{eqnarray}

\begin{equation}\label{5.2.6}
c_{1}^{\left(i\right)}=\begin{cases}
\sigma\sqrt{\frac{t_{i}}{t}\left(t-t_{i}\right)}-\Phi^{-1}\left(x\right)\;\;\; & i<j,\\
\sigma\sqrt{\left(t_{i}-t\right)}-\Phi^{-1}\left(x\right) & i\geq j.
\end{cases}
\end{equation}
where $x \in \left(0,1\right)$ solves equation \eqref{5.2.5}. The optimal upper bound in this case is then given by minimising equation \eqref{5.2.5a} over $t\in \left[0,T\right]$. As before, invoking the put-call parity of section 2, we have for the Swiss Re bond
\begin{equation}\label{5.2.7}
P \leq \left(\mbox{ub}_{t}^{1}-G\right)^{+}=:\mbox{SWUB}_{t}^{\left(BS\right)},
\end{equation}
where G is defined in \eqref{2.13}.

\subsection{Log Gamma Distribution}
The log Gamma distribution is a particular type of transformed Gamma distribution. The mortality index `$q$' is said to follow log Gamma distribution if
\begin{equation}\label{4.0.3}
\frac{\log_{e}q-\mu}{\sigma}=x\sim Gamma\left(p,a\right),
\end{equation}
where $\mu, \sigma, p$ and $a$ are parameters ($>0$) and $log$ is the natural logarithm. Useful references for reading about transformed gamma distribution are \cite{Johnson2}, \cite{Vitiello} and \cite{Cheng}.
\subsubsection{The Lower Bound $\mbox{SWLB}_{t}^{\left(LG\right)}$}
In this case the marginal cdfs $F_{Y_{i}}$ are strictly increasing. So, for the log-gamma distribution we obtain the following compact expression for ${lb}_{t}^{\left(2\right)}$ and then subtract $G$ from it to obtain $\mbox{SWLB}_{t}^{\left(LG\right)}$.
{
\allowdisplaybreaks
\begin{eqnarray}\label{4.0.10}
\mbox{lb}_{t}^{\left(2\right)}& = & 5Ce^{-rT}\Bigg(\sum_{i=1}^{j-1}q_{0}^{-t_{i}/t}\left(\frac{e^{\frac{t_{i}}{t}\mu}}{\left(\sigma^{"}\right)^{p}}\left[1-G\left(d_{2}^{'},\;p,\sigma^{"}\right)\right]-K_{1}\left[1-G\left(d_{2}^{'},\;p\right)\right]\right) \nonumber\\
& {} {} & {} \;\;\;\;\;\;\;\;+{\displaystyle \sum_{i=j}^{n}}\frac{e^{r\left(t_{i}-t\right)}}{q_{0}}\left(q_{0}e^{rt}\left[1-G\left(d_{1},\;p\right)\right]-K_{2}\left[1-G\left(d_{2},\;p\right)\right]\right)\Bigg)
\end{eqnarray}
}

where we have
\[
\sigma^{"}=1-\sigma^{'}\frac{t_{i}}{t},\; \sigma^{'}=1-\left(q_{0}e^{rt-\mu}\right)^{1/p},\]
\[d_{2}^{'}=\frac{lnd_{1}^{'}-\mu}{\sigma},\; d_{1}^{'}=q_{0}\left(1.3+\left(x^{t_{i}/t}-1.3\right)^{+}\right)^{t/t_{i}},
\]
\[K_{1}=\left(d_{1}^{'}\right)^{t_{i}/t},\;K_{2}=q_{0}\left(\frac{1.3}{e^{r\left(t_{i}-t\right)}}+\left(x-\frac{1.3}{e^{r\left(t_{i}-t\right)}}\right)^{+}\right)\],
\[d_{1}=\frac{lnK_{2}-\mu}{q_{0}e^{rt-\mu}-1},\;d_{2}=d_{1}+lnK_{2}-\mu,\]

\[
G\left(x,p\right)={\displaystyle \int_{0}^{x}\frac{1}{\Gamma\left(p\right)}x^{p-1}e^{-x}dx}
\]
and
\[G\left(x,p,\sigma^{"}\right)={\displaystyle \int_{0}^{x}\frac{\left(\sigma^{"}\right)^{p}}{\Gamma\left(p\right)}x^{p-1}e^{-\left(\sigma^{"}x\right)}dx}.
\]

\subsubsection{The Upper Bound $\mbox{SWUB}_{1}^{\left(LG\right)}$}
The first upper bound given in section 4.1 can be derived in the same manner as above exploiting that the marginal cdfs $F_{q_{i}}$ are strictly increasing. The results are given in Tables 5 and 6.

\section{Numerical Results}
The stage is now set to investigate the applications of the theory derived in the previous sections. We have successfully obtained a number of lower bounds and upper bounds for the Swiss Re bond in sections 3 and 4. In section 5 we have furnished a couple of examples. We now test these vis-a-vis the well-known Monte Carlo estimate for the Swiss Re bond. We assume that $C=1$ in all the examples. We first carry out this working under the well known \cite{Black} model in finance and then for a couple of transformed distributions. The nomenclature for the bounds has already been specified in sections 3 and 4. In all the examples, the marginal cdfs are strictly increasing.

In tables 1 and 2, we assume that the mortality evolution process $\left\{ q_{t}\right\} _{t\geq0}$ obeys the Black-Scholes model, specified by the following stochastic differential equation (SDE)
\[
dq_{t}=rq_{t}dt+\sigma q_{t}dW_{t}.
\]
In order to simulate a path, we will consider the value of the mortality index in the three years that form the term of the bond, i.e., $n=3$. In fact we consider the time points as $t_{1}=1,...,t_{n}=T=3$. We invoke the following equation to generate the mortality evolution:
\begin{equation}\label{4.0.1}
q_{t_{j}}=q_{t_{j-1}}\exp\left[\left(r-\frac{1}{2}\sigma^{2}\right)\delta t+\sigma\sqrt{\delta t}Z_{j}\right]\;\;\; Z_{j}\sim N\left(0,1\right),\;\;\; j=1,2,\ldots,n.
\end{equation}
We highlight below the parameter choices in accordance with \cite{Lin}. The value of the interest rate is varied in table 1 while table 2 experiments with the variation in the base value of the mortality index while assuming a zero interest rate.
Parameter choices for tables 1 and 2 with $t$ specified in terms of years are:
\[
q_{0}=0.008453,\; T=3,\; t_{0}=0,\; n=3, \; \sigma=0.0388.
\]

Table 2 is followed by figures 1-3. While figures 1 and 2 depict comparisons between the bounds, figure 3 portrays the price bounds for the Swiss Re bond generated by the Black-Scholes model. We will let MC denote the Monte Carlo estimate for the Swiss Re bond.

Table 1 reflects that the relative difference ($=\frac{|bound-MC|}{MC}$) between any bound and the benchmark Monte Carlo estimate increases with an increase in the interest rate for a fixed value of the base mortality index $q_{0}$. This observation is echoed by figure 1. On the other hand, figure 2 depicts the difference between the Monte Carlo estimate of the Swiss Re bond and the derived bounds. The bound $\mbox{SWLB}_{t}^{\left(BS\right)}$ fares much better than $\mbox{SWLB}_{1}$. The absolute difference between the estimated price and the bounds increase as the value of the base mortality index is increased and then there is a switch and this gap begins to diminish. This observation is supported by the fact that an increase in the starting value of mortality increases the possibility of a catastrophe which leads to the washing out of the principal or in other words the option goes out of money.

We now consider an additional example. Assume that the mortality rate `$q$' obeys the four-parameter transformed Normal ($S_{u}$) Distribution (for details see \cite{Johnson} and \cite{Johnson2}) which is defined as follows

\begin{equation}\label{4.0.2}
sinh^{-1}\left(\frac{q-\alpha}{\beta}\right)=x\sim N\left(\mu,\sigma^{2}\right),
\end{equation}
where $\alpha, \beta, \mu$ and $\sigma$ are parameters ($\beta, \sigma > 0$) and $sinh^{-1}$ is the inverse hyperbolic sine function.

\begin{table} [h!]
\scriptsize
\begin{center}
\begin{threeparttable}
\centering
\begin{tabular}{|r||r|r|r|r|r|r|r|r|r|}
\hline
$\;$r$\;$ & $\mbox{SWLB}_{0}\;\;\;\;$ & $\mbox{SWLB}_{1}\;\;\;\;$ & $\;\;\mbox{SWLB}_{t}^{\left(BS\right)}\;\;$ &\; $MC\; \mbox{with S.E.}\;\;\;\;$ & $\;\;\mbox{SWUB}_{t}^{\left(BS\right)}\;\;$ & $\mbox{SWUB}_{1}\;\;\;$ \\ \hline\hline
0.035 & 0.899130889131 & 0.899130889153 & 0.899131577419 & 0.899131338643 & 0.899131588500 & 0.899131637780 \\
      &                &                &                & (0.000007814868) &                &                \\
0.030 & 0.913324024542 & 0.913324024546 & 0.913324256506 & 0.913324365180 & 0.913324317265 & 0.913324320930 \\
      &                &                &                & (0.000005483857) &                &                \\
0.025 & 0.927447505802 & 0.927447505803 & 0.927447580428 & 0.927447582074 & 0.927447605312 & 0.927447619324 \\
      &                &                &                & (0.000003766095) &                &                \\
0.020 & 0.941626342686 & 0.941626342687 & 0.941626365600 & 0.941626356704 & 0.941626369727 & 0.941626384749 \\
      &                &                &                & (0.000002549695) &                &                \\
0.015 & 0.955935721003 & 0.955935721003 & 0.955935727716 & 0.955935715489 & 0.955935732230 & 0.955935736078 \\
      &                &                &                & (0.000001673442) &                &                \\
0.010 & 0.970419124546 & 0.970419124546 & 0.970419126422 & 0.970419112046 & 0.970419126802 & 0.970419129772 \\
      &                &                &                & (0.000001032941) &                &                \\
0.005 & 0.985101139986 & 0.985101139986 & 0.985101140486 & 0.985101142704 & 0.985101140840 & 0.985101141738 \\
      &                &                &                & (0.000000646744) &                &                \\
0.000 & 0.999995778016 & 0.999995778016 & 0.999995778143 & 0.999995770298   & 0.999995778175 & 0.999995778584 \\
      &                &                &                & (0.000000405336) &                &                \\
 \hline
\end{tabular}
\caption{Lower Bounds and Upper Bound $\mbox{SWUB}_{1}$ for the Swiss Re Mortality Bond under the Black-Scholes Model with $q_{0}=0.008453$ and $\sigma=0.0388$ in accordance with \cite{Lin}. MC Simulations:5000000 iterations (Antithetic Method)}
\end{threeparttable}
\end{center}
\end{table}

\begin{table} [h!]
\scriptsize
\begin{center}
\begin{threeparttable}
\centering
\begin{tabular}{|r||r|r|r|r|r|r|r|r|}
\hline
$\mbox{q}_{0}\;$ & $\mbox{SWLB}_{0}\;\;\;\;$ & $\mbox{SWLB}_{1}\;\;\;\;$ & $\;\;\mbox{SWLB}_{t}^{\left(BS\right)}\;\;$ & $\;MC\; \mbox{with S.E.}\;\;\;\;\;\;$ & $\;\;\mbox{SWUB}_{t}^{\left(BS\right)}\;\;$ & $\mbox{SWUB}_{1}\;\;\;$ \\ \hline\hline
0.007 & 1.000000000000 & 1.000000000000 & 1.000000000000 & 1.000000000000 & 1.000000000000 & 1.000000000000 \\
      &                &                &                & (0.000000000000) &              &                \\
0.008 & 0.999999915252 & 0.999999915252 & 0.999999915252 & 0.999999915033 & 0.999999915253 & 0.999999915253 \\
      &                &                &                & (0.000000052478) &                &                \\
0.008453 & 0.999995778016 & 0.999995778016 & 0.999995778143 & 0.999995770298 & 0.999995778175 & 0.999995778584 \\
      &                &                &                & (0.000000405336) &                &                \\
0.009 & 0.999821987943 & 0.999821987950 & 0.999822025863 & 0.999822630214 & 0.999822374801 & 0.999822875816 \\
      &                &                &                & (0.000003051524) &                &                \\
0.010 & 0.978292691035 & 0.978310383929 & 0.978503560221 & 0.978782997810 &	0.978292691184 & 0.986262918347 \\
      &                &                &                & (0.000042738093) &                &                \\
0.011 & 0.572750782004 & 0.610962124258	& 0.610962123857 & 0.652245039892 & 0.572755594265 & 0.877336305502 \\
      &                &                &                & (0.000090193709) &                &                \\
0.012 & 0.000000000000 & 0.040209774144 & 0.040209770810 & 0.094677358603 &	0.000000000000 & 0.395672911251 \\
      &                &                &                & (0.000089559585) &                &                \\
0.013 & 0.000000000000 & 0.000000000000 & 0.000000000000 & 0.001665407936 & 0.000000000000 & 0.083466184427 \\
      &                &                &                & (0.000011391823) &                &                \\
0.014 & 0.000000000000 & 0.000000000000 & 0.000000000000 & 0.000002890238 &	0.000000000000 & 0.008942985848 \\
      &                &                &                & (0.000000379522) &                &                \\
\hline
\end{tabular}
\caption{Lower Bounds and Upper Bound $\mbox{SWUB}_{1}$ for the Swiss Re Mortality Bond under the Black-Scholes Model with $r=0.0$ and $\sigma=0.0388$ in accordance with \cite{Lin}. MC Simulations:5000000 iterations (Antithetic Method)}
\end{threeparttable}
\end{center}
\end{table}

\newpage

\begin{figure}[H]
    \centering
     \includegraphics[clip, trim=2cm 18.5cm 0.9cm 2cm, width=1.00\textwidth]{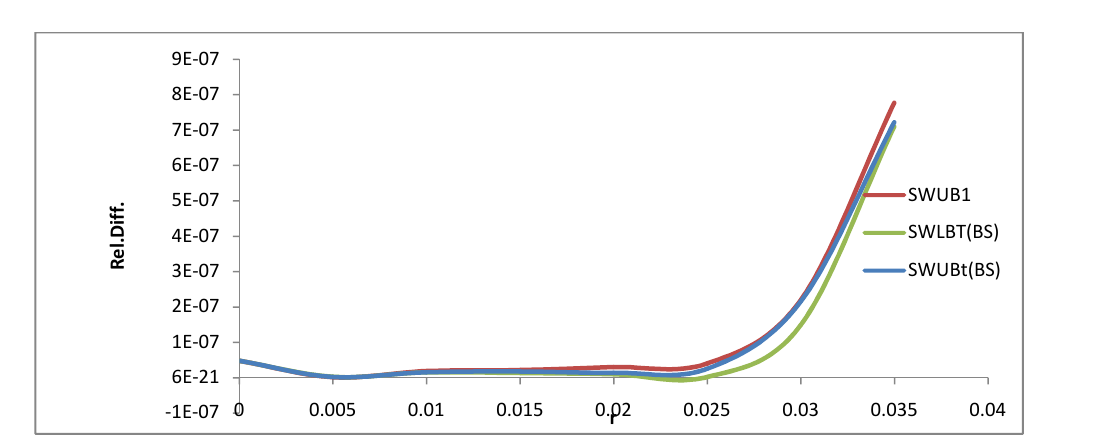}
      \caption{Relative Difference of $\mbox{SWLB}_{t}^{\left(BS\right)}$, $\mbox{SWUB}_{t}^{\left(BS\right)}$ and $\mbox{SWUB}_{1}$ w.r.t. MC estimate under Black-Scholes model}
      \label{fig40}
\end{figure}

\begin{figure}[H]
    \centering
     \includegraphics[clip, trim=2cm 18.3cm 0.9cm 2.1cm, width=1.00\textwidth]{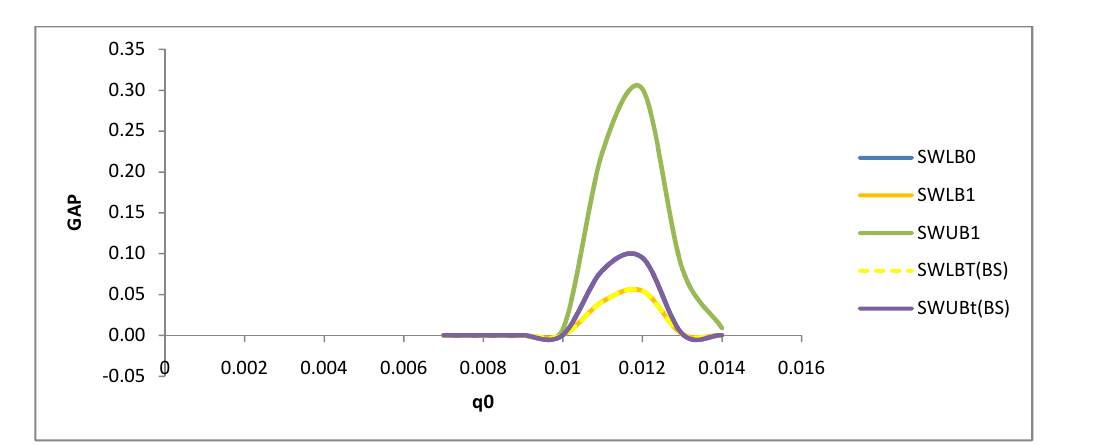}
      \caption{Comparison of different bounds under B-S model in terms of difference from MC estimate for r=0}
      \label{fig41}
\end{figure}

\begin{figure}[H]
    \centering
     \includegraphics[clip, trim=2cm 16.7cm 0.9cm 2cm, width=1.00\textwidth]{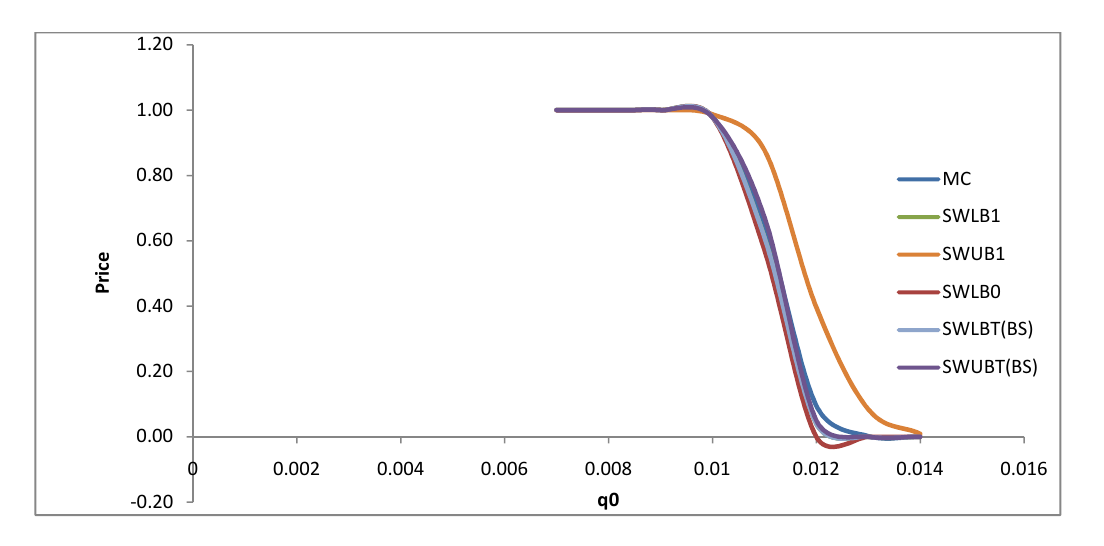}
      \caption{Price Bounds under Black-Scholes model for the parameter choice of Lin and Cox(2008) Model}
      \label{fig42}
\end{figure}

For table 3, we vary the interest rate as in table 1 and use the parameter set employed by \cite{Tsai}. The aforesaid authors use the mortality catastrophe model of \cite{Lin} to generate the data and then utilize the quantile-based estimation of \cite{Slifker} to estimate the parameters of the $S_{u}$-fit. The initial mortality rate and time points are same as for tables 1 and 2. The following arrays present the values of the parameters for the three years 2004, 2005 and 2006 that were covered by the Swiss Re bond.
\[
\alpha=[0.008399, 0.008169, 0.007905],\;\beta=[0.000298, 0.000613, 0.000904],
\]
\[
\mu=[0.70780, 0.58728, 0.58743]\;\text{and}\;\sigma=[0.67281, 0.50654, 0.42218].
\]

The value of $\mbox{SWLB}_{t}^{\left(2\right)}$ in table 3 has been calculated by using `Numerical Integration' in MATLAB since the first term in \eqref{4.5.10} can not be calculated mathematically. Table 3 adds weight to the claim that the bounds are extremely tight for a large class of models assuming a variety of distributions for the mortality index. Finally in tables 4 and 5, we experiment with log gamma distribution by varying the interest rate in table 4 and the base mortality rate in the the latter. The parameters are chosen as in \cite{Cheng} who employ an approach similar to \cite{Tsai} outlined above with $q_{0}=.0088$ but use maximum likelihood estimation to obtain the parameters of the fitted log gamma distribution. As before, the following arrays present the year wise parameters
\[
p=[61.6326, 64.2902, 71.8574],\;a=[0.0103, 0.0098, 0.0080],
\]
\[
\mu=[-5.2452, -5.4600, -5.7238]\;\text{and}\;\sigma=[7.4\times10^{-5}, 9.5\times10^{-5}, 9.4\times10^{-5}].
\]
Tables 4 and 5 clearly shows that even for non-normal universe, the bounds are extremely precise. Figures 4-6 are drawn on the lines of figures 1-3 and strongly support our observation.

\begin{table} [h!]
\footnotesize
\begin{center}
\begin{threeparttable}
\centering
\begin{tabular}{|r||r|r|r|r|r|r|r|r|r|r|}
\hline
$\;$r$\;$ & $\mbox{SWLB}_{0}\;\;\;\;$ & $\mbox{SWLB}_{1}\;\;\;\;$ & $\;\mbox{SWLB}_{t}^{\left(2\right)}\;\;\;\;\;$ & $\;MC\;\;\;\;\;$ & $\;\;\mbox{S.E.(MC)}\;\;\;\;$ & $\mbox{SWUB}_{1}\;\;\;$ \\ \hline\hline
0.035 & 0.88325546 & 0.88432143 & 0.88554815 & 0.88468962 &	0.00006349 & 0.88680657 \\
0.030 & 0.90340398 & 0.90401002 & 0.90469396 & 0.90422765 & 0.00004987 & 0.90548179 \\
0.025 & 0.92160707 & 0.92193552 & 0.92229117 & 0.92201394 &	0.00003804 & 0.92275950 \\
0.020 & 0.93840783 & 0.93857698 & 0.93874756 & 0.93863396 &	0.00002794 & 0.93901043 \\
0.015 & 0.95428713 & 0.95436972 & 0.95444409 & 0.95441569 &	0.00001956 & 0.95458265 \\
0.010 & 0.96963954 & 0.96967776 & 0.96970660 & 0.96968765 &	0.00001352 & 0.96977488 \\
0.005 & 0.98476274 & 0.98477952 & 0.98478912 & 0.98478917 &	0.00000859 & 0.98482046 \\
0.000 & 0.99986135 & 0.99986838 & 0.99987088 & 0.99987622 &	0.00000513 & 0.99988427 \\
\hline
\end{tabular}
\caption{Lower Bounds and Upper Bound $\mbox{SWUB}_{1}$ for the Swiss Re Mortality Bond under the $S_{u}$ distribution with $q_{0}=0.008453$ and parameter choice in accordance with \cite{Tsai}. MC Simulations:2000000 iterations (Antithetic Method)}
\end{threeparttable}
\end{center}
\end{table}

\begin{table} [h!]
\footnotesize
\begin{center}
\begin{threeparttable}
\centering
\begin{tabular}{|r||r|r|r|r|r|r|r|r|r|}
\hline
$\;$r$\;$ & $\mbox{SWLB}_{0}\;\;\;\;$ & $\mbox{SWLB}_{1}\;\;\;\;$ & $\;\;\mbox{SWLB}_{t}^{\left(LG\right)}\;\;$ & $\;MC\;\;\;\;\;\;\;$ & $\;\;\mbox{S.E.(MC)}\;\;\;$ & $\mbox{SWUB}_{1}\;\;\;\;$ \\ \hline\hline
0.035 & 0.84803277 & 0.84842404 & 0.85596973 & 0.85408651 &	0.00049859 & 0.86610436 \\
0.030 & 0.87357702 & 0.87381345 & 0.87911092 & 0.87815608 & 0.00044050 & 0.88724013 \\
0.025 & 0.89710281 & 0.89724267	& 0.90088166 & 0.90050920 &	0.00038741 & 0.90728309 \\
0.020 & 0.91889696 & 0.91897792	& 0.92142119 & 0.92103020 & 0.00034012 & 0.92636640 \\
0.015 & 0.93924097 & 0.93928679	& 0.94088833 & 0.94092949 &	0.00028650 & 0.94463331 \\
0.010 & 0.95840372 & 0.95842907	& 0.95945270 & 0.95947457 &	0.00024259 & 0.96223065 \\
0.005 & 0.97663543 & 0.97664912 & 0.97728623 & 0.97748291 &	0.00020357 & 0.97930297 \\
0.000 & 0.99416285 & 0.99417007	& 0.99455565 & 0.99466024 &	0.00016677 & 0.99598733 \\
\hline
\end{tabular}
\caption{Lower Bounds and Upper Bound $\mbox{SWUB}_{1}$ for the Swiss Re Mortality Bond under the transformed gamma distribution with $q_{0}=0.0088$ and parameter choice in accordance with \cite{Cheng}. MC Simulations:100000 iterations}
\end{threeparttable}
\end{center}
\end{table}

\begin{table} [h!]
\footnotesize
\begin{center}
\begin{threeparttable}
\centering
\begin{tabular}{|r||r|r|r|r|r|r|r|r|}
\hline
$\mbox{q}_{0}\;$ & $\mbox{SWLB}_{0}\;\;\;\;$ & $\mbox{SWLB}_{1}\;\;\;\;$ & $\mbox{SWLB}_{t}^{\left(LG\right)}\;\;\;$ & $\;MC\;\;\;\;\;$ &  $\;\;\mbox{S.E.(MC)}\;\;\;$ & $\mbox{SWUB}_{1}\;\;\;\;$ \\ \hline\hline
0.008 & 0.99976607 & 0.99976607	& 0.99977284 & 0.99978465 & 0.00003227 & 0.99977956 \\
0.0088 & 0.99416285 & 0.99417007 &	0.99455565	& 0.99466024 & 0.00016677 & 0.99598733 \\
0.009 &	0.98910499 & 0.98914615	& 0.98995211 & 0.99003596 &	0.00023335 & 0.99338335 \\
0.010 &	0.87669254 & 0.88804918	& 0.89637631 & 0.89137680 &	0.00077924 & 0.95818959 \\
0.011 &	0.41097106 & 0.59608967	& 0.59608967 & 0.56844674 &	0.00128761 & 0.83720797 \\
0.012 &	0.00000000 & 0.27104597	& 0.27104597 & 0.20822580 &	0.00105003 & 0.61383872 \\
0.013 &	0.00000000 & 0.08274071	& 0.08274071 & 0.04612178 &	0.00052388 & 0.38182244 \\
0.014 &	0.00000000 & 0.01270202	& 0.01270202 & 0.00673234 &	0.00019165 & 0.21222938 \\
0.015 &	0.00000000 & 0.00000000	& 0.00000000 & 0.00084831 &	0.00006528 & 0.11042035 \\
0.016 &	0.00000000 & 0.00000000	& 0.00000000 & 0.00009165 &	0.00002235 & 0.05553927 \\
0.017 &	0.00000000 & 0.00000000	& 0.00000000 & 0.00000621 &	0.00000447 & 0.02757685 \\
0.018 &	0.00000000 & 0.00000000	& 0.00000000 & 0.00000205 &	0.00000145 & 0.01369796 \\
\hline
\end{tabular}
\caption{Lower Bounds and Upper Bound $\mbox{SWUB}_{1}$ for the Swiss Re Mortality Bond under the transformed gamma distribution with $r=0.0$ and parameter choice in accordance with \cite{Cheng}. MC Simulations:100000 iterations}
\end{threeparttable}
\end{center}
\end{table}

\newpage

\begin{figure}[H]
    \centering
     \includegraphics[clip, trim=2cm 18.3cm 0.9cm 2.2cm, width=1.00\textwidth]{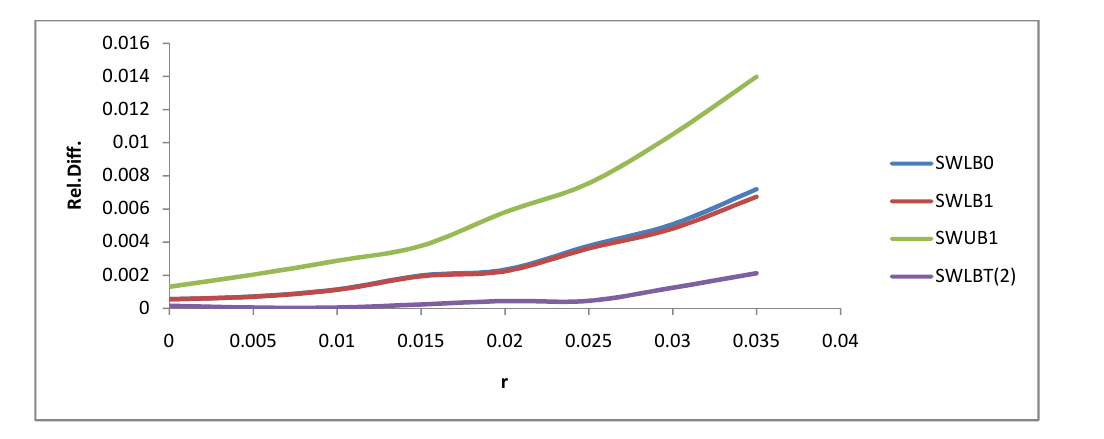}
      \caption{Relative Difference of Lower Bounds and SWUB1 w.r.t. MC estimate under Transformed Gamma Distribution}
      \label{fig43}
\end{figure}

\begin{figure}[H]
    \centering
     \includegraphics[clip, trim=2cm 18.3cm 0.9cm 2.2cm, width=1.00\textwidth]{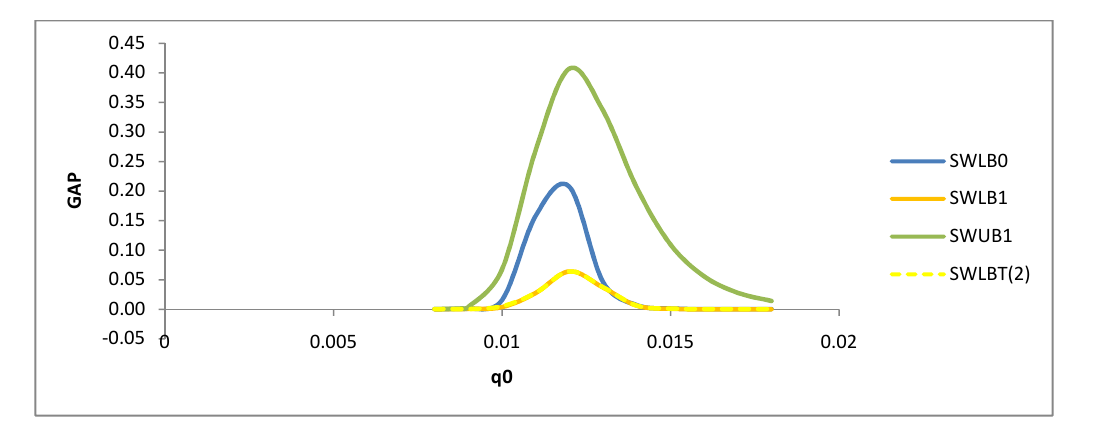}
      \caption{Comparison of different bounds under Transformed Gamma Distribution in terms of difference from MC estimate for r=0}
      \label{fig44}
\end{figure}

\begin{figure}[H]
    \centering
     \includegraphics[clip, trim=2cm 16.7cm 0.9cm 2.2cm, width=1.00\textwidth]{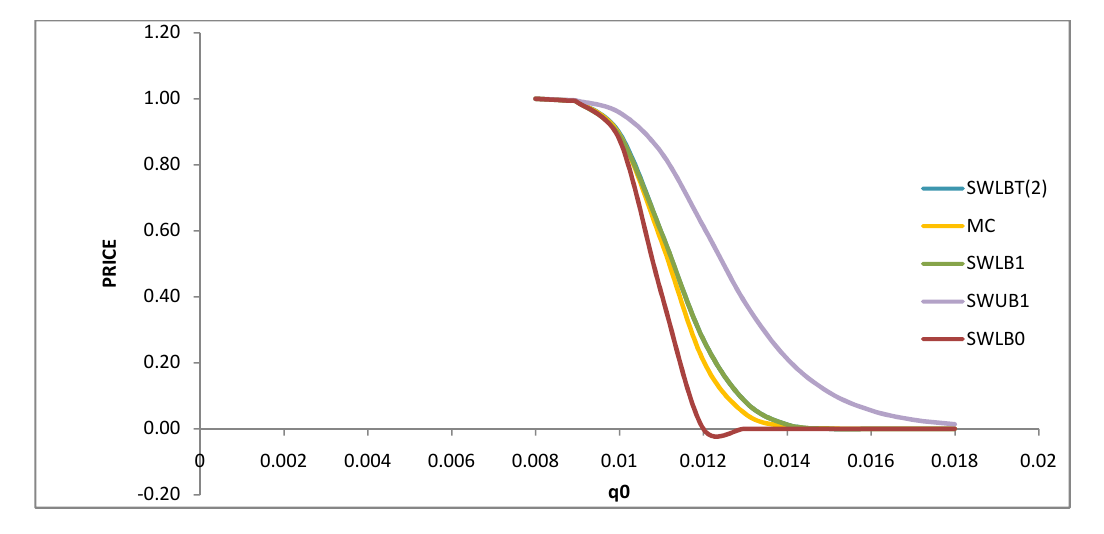}
      \caption{Price Bounds under Transformed Gamma Distribution for the parameter choice of Lin and Cox(2008) Model}
      \label{fig45}
\end{figure}

\section{Conclusions}
Mortality forecasts are extremely significant in the management of life insurers and private pension plans. Securitization and construction of mortality bonds has become an important part of capital market solutions. Prior to the launch of the Swiss Re bond in 2003, life insurance securitization was not designed to handle mortality risk.

This article investigates the designing of price bounds for the Swiss Re mortality bond 2003. As stated in \cite{Deng}, an incomplete mortality market that has no arbitrage opportunities guarantees the existence of at least one risk-neutral measure termed the equivalent martingale measure $Q$ that can be used for calculating the fair prices of mortality securities. We rely on this fact and devise bounds for the mortality security in question without assuming any particular model. Model-specific bounds can then be achieved by plugging in the requisite models into the general bounds. 

To the best of our knowledge, there is only one earlier publication by \cite{Huang} in direction of price bounds for the Swiss Re bond. However, these authors propose gain-loss bounds that suffer from model risk. Our results assume the trading of vanilla options written on the mortality index, as in that case one can use the market price of these options to create bounds which are truly model independent. A worthy observation is that the stimulant for the present work is the theory of comonotonicity. One can therefore easily extend this approach for computing tight bounds for other mortality and longevity linked securities.

\medskip

\noindent\textbf{Acknowledgments.} R.B. gratefully acknowledges the financial support of the Institute and Faculty of Actuaries..

\end{document}